\UseRawInputEncoding
\documentclass[conference]{IEEEtran}
\IEEEoverridecommandlockouts

\usepackage{amsthm}
\usepackage{amsmath}
\usepackage{amsfonts}
\usepackage[linesnumbered,ruled]{algorithm2e}
\usepackage{graphicx, float, subfig}
\usepackage{booktabs}
\usepackage[figuresright]{rotating}
\usepackage{makecell}
\usepackage{hyperref}
\usepackage{booktabs}
\usepackage{float}
\usepackage{multirow}
\usepackage{array}
\usepackage{balance}

\newtheorem{definition}{Definition}
\newtheorem{theorem}{Theorem}
\newtheorem{lemma}{Lemma}

\def\BibTeX{{\rm B\kern-.05em{\sc i\kern-.025em b}\kern-.08em
    T\kern-.1667em\lower.7ex\hbox{E}\kern-.125emX}}

\makeatletter
\newcommand{\linebreakand}{%
  \end{@IEEEauthorhalign}
  \hfill\mbox{}\par
  \mbox{}\hfill\begin{@IEEEauthorhalign}
}
\makeatother

\begin{document}
\title{K-stars LDP: A Novel Framework for $(p,q)$-clique Enumeration under Local Differential Privacy}

\author{
    \IEEEauthorblockN{
    Henan Sun$^\dagger$,
    Zhengyu Wu$^\dagger$,
    Rong-Hua Li$^\dagger$, 
    Guoren Wang$^\dagger$,
    Zening Li$^\dagger$}
    \IEEEauthorblockA{
    $^\dagger$ Beijing Institute of Technology, Beijing, China}
    \IEEEauthorblockA{
    magneto0617@foxmail.com,
    Jeremywzy96@outlook.com,\\
    lironghuabit@126.com,
    wanggrbit@gmail.com,
    zening-li@outlook.com
    }
}

\maketitle

\begin{abstract}
$(p,q)$-clique enumeration on a bipartite graph is critical for calculating clustering coefficient and detecting densest subgraph. It is necessary to carry out subgraph enumeration while protecting users’ privacy from any potential attacker as the count of subgraph may contain sensitive information. Most recent studies focus on the privacy protection algorithms based on edge LDP (Local Differential Privacy). However, these algorithms suffer a large estimation error due to the great amount of required noise. In this paper, we propose a novel idea of $k$-stars LDP and a novel $k$-stars LDP algorithm for $(p,q)$-clique enumeration with a small estimation error, where a $k$-stars is a star-shaped graph with $k$ nodes connecting to one node. The effectiveness of edge LDP relies on its capacity to obfuscate the existence of an edge between the user and his one-hop neighbors. This is based on the premise that a user should be aware of the existence of his one-hop neighbors. Similarly, we can apply this premise to $k$-stars as well, where an edge is a specific genre of $1$-stars. Based on this fact, we first propose the $k$-stars neighboring list to enable our algorithm to obfuscate the existence of $k$-stars with Warner’s RR. Then, we propose the absolute value correction technique and the $k$-stars sampling technique to further reduce the estimation error. Finally, with the two-round user-collector interaction mechanism, we propose our $k$-stars LDP algorithm to count the number of $(p,q)$-clique while successfully protecting users' privacy. Both the theoretical analysis and experiments have showed the superiority of our algorithm over the algorithms based on edge LDP.

\end{abstract}


\section{Introduction}
\label{sec:Intro}

Subgraph counting on a bipartite graph is a fundamental task in graph theory, which expands to many real-world applications, such as online customer-product analysis, author-paper relationship and so on~\cite{wang2021efficient, wang2021discovering, wang2022towards}. For example, given a bipartite graph $G(V=(U,L),E)$, a $(p,q)$-clique is a complete subgraph $B(X,Y)$ of the bipartite graph G, where $ X\subseteq U $, $ |X|=p $, $ Y\subseteq L $, $ |Y|=q $, and $ \forall (u,v)\in X\times Y $, $ (u,v)\in E(G) $, as is shown in Fig. \ref{fig:pq-clique}. For example, when $p=2$ and $q=2$, the number of $(2,2)$-clique (also known as butterfly) can be used to calculate the clustering coefficient in a bipartite graph. In addition, $(p,q)$-clique can also be used for densest subgraph detection \cite{yang2021p}. However, some sensitive information is often involved in the graph data, which may be leaked from the counting results of $(p,q)$-clique \cite{imola2021locally}.
\begin{figure}[htbp]
    \centering
    \includegraphics[scale=0.35]{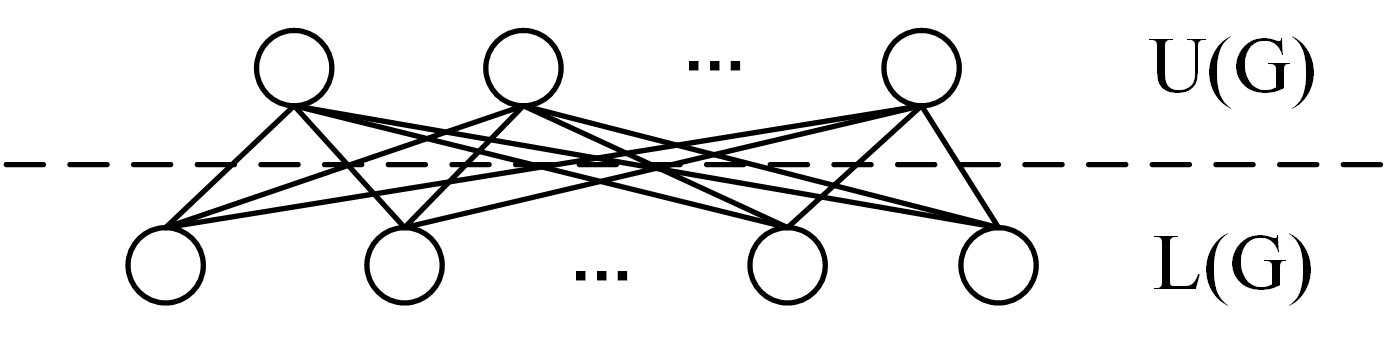}
    \caption{An example of $(p,q)$-clique.}
    \label{fig:pq-clique}
\end{figure}  
In order to analyze the structure of graphs while protecting users' privacy, DP (Differential Privacy) has been widely used as a means of privacy protection \cite{dwork2008differential, ding2021differentially, ficek2021differential, zhao2022survey, girgis2021shuffled}. Because of its capacity to protect users' data from attackers with arbitrary background knowledge, DP has become the gold standard for data privacy \cite{bi2022distribution}. DP can be divided into CDP (Central DP) and LDP (Local DP) \cite{karwa2011private, bi2020privacy, ye2020lf}. CDP assumes that users directly send their data to a trusted data collector who then analyzes and publishes all users' data with the added noises. In contrast, LDP assumes an untrustworthy-collector scenario where users perturb their data before sending it to the collector. Due to the risk of illegal access or internal fraud under the assumption of CDP, LDP is a safer way for privacy protection and has attracted more attention from researchers \cite{bebensee2019local, hidano2022degree, yang2020local}. 

The existing privacy protection algorithms for subgraph counting are usually under edge LDP \cite{kasiviswanathan2011can,imola2021locally, imola2022communication}. Users send the obfuscated edges to the collector based on Warner's RR (Randomized Response) \cite{warner1965randomized}, and then the collector sends the noisy edges back to the corresponding users to count the noisy number of the subgraph. Unfortunately, because of the ignorance of the structure information of the subgraph, these algorithms require too much noise to protect data privacy, which leads to a prohibitively large estimation error.

We observe a basic fact: to achieve effective privacy protection, traditional edge LDP obfuscates the presence of edges around a user, since he can clearly know the existence/non-existence of his one-hop neighbors under the LDP assumption. However, this functionality can also apply to subgraphs such as $2$-stars and $3$-stars, which we refer to as $k$-stars (where edge is also a $1$-stars). Using the numerical value 1 to represent the existence of $k$-stars and 0 to indicate their non-existence, users can obfuscate their presence. Our proposed method achieves this with less noise than traditional edge LDP, resulting in increased data utility. For example, to count $(3,3)$-clique, traditional edge LDP requires 6 noisy edges (namely 6 random variables), as shown in Fig. \ref{fig:33-clique} in the left. However, these 6 noisy edges are equivalent to 2 noisy $3$-stars (namely 2 random variables), as shown in Fig. \ref{fig:33-clique} in the right. The noise of 2 random variables is far less than that of 6 random variables, thus leading to increased data utility. The $k$-stars LDP is proposed based on this fact.
\begin{figure}[htbp]
    \centering
    \includegraphics[scale=0.35]{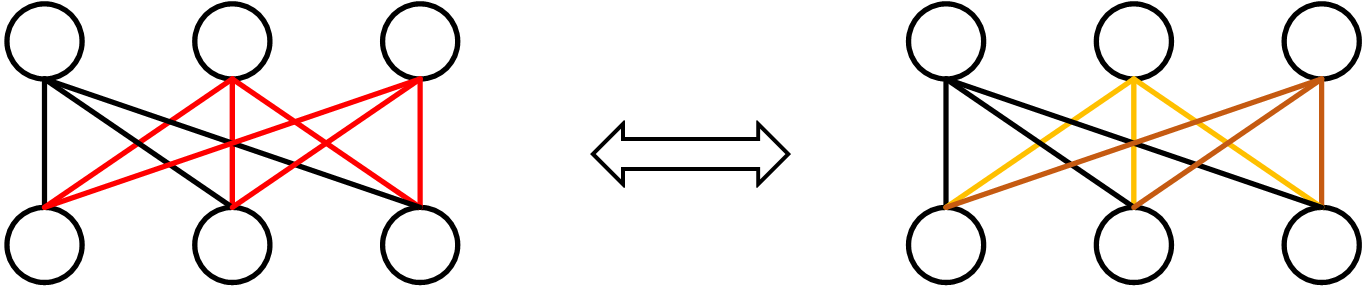}
    \caption{The example of $(3,3)$-clique. To count $(3,3)$-clique, traditional edge LDP algorithms require 6 noisy edges (6 random variables), depicted as red edges on the left. However, these noisy edges can be viewed as 2 noisy $3$-stars (2 random variables), depicted as yellow one and brown one on the right.}
    \label{fig:33-clique}
\end{figure}  

\noindent
\textbf{Our contributions.} The specific contributions can be summarized as follows: 

\begin{itemize}
\item[$\bullet$] \textbf{Novel framework}. We notice a basic fact that the effectiveness of traditional edge LDP algorithm is the ability of obfuscating the existence of $k$-stars, as edge also belongs to one of the $k$-stars. In other words, $k$-stars can be the basic obfuscated unit as edge does. Meanwhile, $k$-stars requires far less noise than edge does, which is proved in Theorem \ref{theo:utility}. Based on this fact, we propose the novel framework named $k$-stars LDP. 
\item[$\bullet$] \textbf{New algorithm}. Based on two-round user-collector interaction mechanism and RR, we calculate the noisy number of $(p,q)$-clique under the $k$-stars neighboring list. To the best of our knowledge, we are the first to study $(p,q)$-clique enumeration under local differential privacy. The theoretical analysis proves the privacy effectiveness, unbiasedness and the utility of our algorithm. 
\item[$\bullet$] \textbf{New mechanisms}. In order to reduce the estimation error, we further proposed two novel techniques: absolute value correction and $k$-stars sampling. Previous works \cite{imola2021locally, imola2022communication} neglected the negative-output problem occurred in the two-round user-collector interaction mechanism, which leads to the negative estimation of subgraph. To address this challenge, we propose the absolute value correction technique. We also propose the $k$-stars sampling technique based on traditional edge sampling to further enhance the utility of algorithms. Both theoretical analysis and experiments have proved the effectiveness of these two techniques. 
\item[$\bullet$] \textbf{SOTA Performance}. We conduct experiments on the four public datasets to evaluate the performance of our algorithm. The results show that our algorithm outperforms the edge LDP algorithm in every aspect evaluated in the experiments.
\end{itemize}

\noindent
\textbf{Organization.} The rest of this paper is organized as follows: Section \ref{sec:RW} gives a summary about recent related works. Section \ref{sec:Pre} introduces the preliminaries of the algorithms. Section \ref{sec:LD} introduces both edge LDP algorithm and $k$-stars LDP algorithm for $(p,q)$-clique enumeration. Section \ref{sec:exp} reports the experimental results of the algorithms. Section \ref{sec:con} concludes this paper.

\section{Related Work}
\label{sec:RW}
\noindent
\textbf{Subgraph counting.} Subgraph counting algorithms on bipartite graphs have been widely studied in the "non-private" mode, such as triangle, butterfly and $(p,q)$-clique, because they are often resource-consuming and time-consuming on large graphs \cite{arifuzzaman2013patric, bera2020degeneracy, bera2020count, eden2017approximately, chu2011triangle, atminas2012linear, dawande1996biclique, chen2022efficient, lu2020biclique, yang2021p}. Edge sampling, as one of the most basic techniques, is widely adopted to improve the utility and scalability of the algorithms \cite{bera2020count, eden2017approximately, tsourakakis2009doulion, wu2016counting}. Wu et al. \cite{wu2016counting} claim that edge sampling is superior than other sampling techniques, such as node sampling. Therefore, this paper proposes $k$-stars sampling based on edge sampling to improve the scalability and utility of the algorithm.\\

\noindent
\textbf{CDP on graphs.} CDP has been widely adopted as a privacy standard in graph data protection \cite{chen2019publishing, day2016publishing, kasiviswanathan2013analyzing, ding2021differentially, zhang2015private}. Zhang et al. \cite{zhang2015private} proposes a ladder framework to specify a probability distribution over possible outputs which aims to maximize the utility for the given input while providing the required CDP privacy level. Ding et al. \cite{ding2021differentially} proposes a graph projection method to reduce the sensitivity caused by a large graph and to satisfy node-level CDP. However, all of the algorithms aforementioned suffer from the risk of illegal access or internal fraud, due to the trusted collector assumption of CDP.\\

\noindent
\textbf{LDP on graphs.} LDP on graphs has been studied in some recent works \cite{ye2020towards, ye2020lf, sun2019analyzing, imola2021locally,imola2022communication}. Sun et al. \cite{sun2019analyzing} proposed a multi-stage mechanism to solve problems of triangle subgraph and k-clique enumeration in the extended local view scenario with the optimized sensitivity and reduced estimation error. Imola et al. \cite{imola2022communication} proposed an algorithm to solve the triangle counting problem based on the multi-stage mechanism. It adopts edge sampling after RR (Randomized Response), "4-cycle strategy" and "double clipping technique" to reduce the estimation error and the communication cost between users and collector.  However, the above algorithms are all based on edge LDP, which ignore the structure information of the subgraph and require too much noise, leading to a large estimation error.
\section{Preliminaries}
\label{sec:Pre}
\subsection{Notations}
Let $\mathbb{N}$, $\mathbb{R}$ and $\mathbb{Z}$ be the natural numbers, real numbers and integers, respectively. For $ i\in \mathbb{N} $, let $[i]$ be the natural number sets from 1 to i, namely $[i]=\left\{1,2,\cdots ,i\right\}$.

Let $G=(V=(U,L),E)$ be an undirected bipartite graph, where U(G) and L(G) represent the sets of  nodes in the upper layer and lower layer respectively, and E(G) be the set of edges in the bipartite graph G. Let n be the number of nodes in V and $v_{i}\in V$ be the i-th node in V ($1\le i\le n$). Notice that there is no edge between nodes both in the upper layer or both in the lower layer, namely $U\times U=\emptyset$, $L\times L=\emptyset$. Let $\mathcal{G}$ be a set of graphs and $f_{pq}:\mathcal{G}\rightarrow \mathbb{N}$ be the function of $(p,q)$-clique enumeration that take $G\in\mathcal{G}$ as input and outputs the count of $(p,q)$-clique in G.

In the scenario of traditional edge LDP, let $\textbf{A}=(a_{i,j})\in \left\{ 0,1\right\}^{n\times n}$ be the adjacency matrix where $a_{i,j}$ be the edge between users $v_{i}$ and $v_{j}$. An edge neighboring list $\textbf{a}_{i}$ is a row in the adjacency matrix \textbf{A}, which is employed in edge LDP algorithms $\mathcal{R}$ \cite{imola2022communication}. 

\subsection{Local Differential Privacy on Graphs}
\label{sec:LDP definition}
Local DP (LDP) protects user privacy by obfuscating data when users submit it to the collector \cite{kasiviswanathan2011can,ye2020lf,imola2022communication}. It prevents collector from accessing users’ real data, thus avoiding the risk of internal fraud or invalid access. Edge LDP is the traditional way to apply LDP to graph. It protects the edge between users so that any attacker cannot distinguish between such two users, one with the edge and the other without. Another alternative, named node LDP, hides the existence of a user, as well as all of his edges \cite{hay2009accurate, zhang2020differentially}. However, in the scenario of subgraph counting, users need to send their data to the collector (whether they are noisy or not). That means the collector has the knowledge of all users' ID, which violates the requirement of node LDP. Thus, we believe edge LDP is more appropriate to the problem of subgraph counting, as \cite{zhang2015private, sun2019analyzing, ye2020towards, ye2020lf, imola2021locally, imola2022communication} do. The specific definition of edge LDP is as follows: 

\begin{definition}[$\epsilon$-edge LDP \cite{imola2022communication}]
Let $n\in \mathbb{N}$, $i\in [n]$, $\epsilon \in \mathbb{R}$. $\mathcal{R}_{i}$ provides $\epsilon$-edge LDP if for any two neighboring list $\textbf{a}_i, \textbf{a}_{i}^{’}\in \left\{0,1\right\}^n$ that differ in one position and $t\in Range(\mathcal{R}_{i})$:
\begin{align}
Pr[\mathcal{R}_{i}(\textbf{a}_{i})=t]\le e^{\epsilon}Pr[\mathcal{R}_{i}(\textbf{a}_{i}^{'})=t]. \tag{1}
\end{align}
\end{definition}

However, for complex subgraphs, such as butterfly, $(2,3)$-clique, $(3,3)$-clique etc., applying edge LDP often requires a great amount of noise, which generates the low data utility.  In fact, as we illustrate in Section \ref{sec:Intro}, $k$-stars can also be directly obfuscated as edges. Therefore, we can extend edge LDP to $k$-stars LDP, which protects user's $k$-stars so that any attacker cannot distinguish such two users, one with the $k$-stars and the other without. In order to apply RR to $k$-stars LDP, we first define the $k$-stars neighboring list, similar to edge LDP: 

\begin{definition}[$k$-stars neighboring list]
Let $KS_{ij_{1}\cdots j_{k}}$ be the $k$-stars centering on user $v_{i}$ and surrounded by users $v_{j_{1}}$ to $v_{j_{k}}$. $KS_{ij_{1}\cdots j_{k}}=1$ if such $k$-stars exists. Otherwise $KS_{ij_{1}\cdots j_{k}}=0$. $k$-stars neighboring list $\textbf{KS}_i$ is the list composed of such $k$-stars centering on user $v_i$.
\end{definition}

In fact, $k$-stars neighboring list is a local projection of edge neighboring list since a user can clearly know the existence of his $k$-stars, as well as his edges. With $k$-stars neighboring list, a $k$-stars LDP algorithm can be obtained. The specific definition of $k$-stars LDP is as follows:

\begin{definition}
Let $n\in \mathbb{N}$, $i\in [n]$, $\epsilon\in \mathbb{R}$. $\mathcal{R}_{i}$ provides $\epsilon$-$k$-stars edge LDP if for any two $k$-stars neighboring list $\textbf{KS}_i$, $\textbf{KS}_{i}^{’}\in\left\{0,1\right\}^n$ that differ in one position and $t\in Range(\mathcal{R}_i)$:
\begin{align}
Pr[\mathcal{R}_{i}(\textbf{KS}_{i})=t]\le e^{\epsilon}Pr[\mathcal{R}_{i}(\textbf{KS}_{i}^{'})=t]. \tag{2}
\end{align}
\end{definition}
 
\subsection{Two-round User-collector Interaction Mechanism}
Two-round user-collector interaction mechanism is applied in some recent research and achieves a better performance than one round user-collector interaction\cite{imola2021locally,imola2022communication}. 
\begin{figure}[htbp]
    \centering
    \includegraphics[scale=0.6]{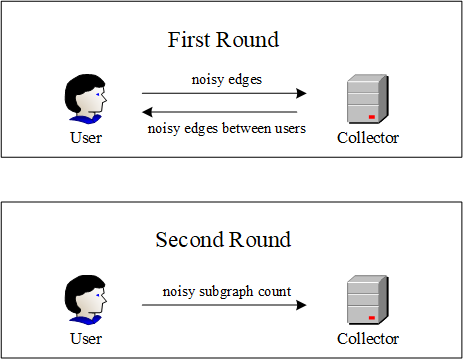}
    \caption{The overview of two-round user-collector interaction mechanism.}
    \vspace{-0.4cm}
    \label{fig:user-collector}
\end{figure} 

As is shown in Fig. \ref{fig:user-collector}, Users obfuscates their data and send it to the collector at the first round. The collector then sends the data back to the corresponding users so that they can enumerate subgraphs locally. Users then calculate the noisy subgraph count and send it to the collector at the second round. At last, the collector aggregates and corrects the noisy subgraph count, obtaining the unbiased estimation. The two-round user-collector interaction mechanism provides LDP guarantee by post-processing property and is applied in this paper.

\subsection{Utility Discussion}
Let $f_{pq}(G)$ be the true count of $(p,q)$-clique and $\tilde{f}_{pq}(G)$ be the unbiased private estimate which satisfies LDP(edge LDP or $k$-stars LDP). The $L_{2}$ loss and relative loss are employed as utility metrics in this paper.

Specifically, $L_2$ loss is the squared error which maps the distance between $f_{pq}(G)$ and $\tilde{f}_{pq}(G)$. As per the bias-variance decomposition theorem, $L_{2}^{2}(f_{pq}(G),\tilde{f}_{pq}(G))=\mathbb{E}(f_{pq}(G)-\tilde{f}_{pq}(G))^2$. Our theoretical analysis employs the expected $L_2$ loss, as \cite{imola2021locally,imola2022communication} do.

However, $L_2$ loss tends to be very large when the subgraph count is large, which is inconvenient to observe the intrinsic properties. As a result, the relative error given by $\frac{|f_{pq}(G)-\tilde{f}_{pq}(G)|}{max(f_{pq}(G),\alpha)}$ is employed in the experiments, where $\alpha \in \mathbb{R}$ is a small value.

\section{Locally Differential Privacy Algorithms for $(p,q)$-clique Enumeration}
\label{sec:LD}

In this section, we first present the baseline algorithm of traditional edge LDP, which employs two-round user-collector mechanism to provide LDP. Then, we introduce the proposed $k$-stars LDP algorithms for $(p,q)$-clique enumeration. Finally, the theoretical performance of both edge LDP and $k$-stars LDP algorithm is analyzed.

\subsection{Baseline Algorithm of Edge LDP}
Inspired by the recent studies \cite{imola2022communication}, we propose the baseline algorithm providing $\epsilon$-edge LDP, as is illustrated below.

Specifically, let $n\in \mathbb{N}$, $i\in [n]$, $\epsilon\in \mathbb{R}$ be the privacy budget, $\textbf{a}_i$ \& $\textbf{a}^{'}_i$ be the edge neighboring lists and $\mathcal{R}_i$ be the local randomized algorithm implemented on user $v_i$. Following the framework of two-round user-collector interaction mechanism, user $v_i$ obfuscates his neighboring list $\textbf{a}_i$ into $\textbf{a}^{'}_{i}$ by RR and sends it to the collector at the first round. Then, the collector sends the noisy neighboring list $\textbf{a}^{'}_{i}$ to the user $v_j$ who needs it to count the noisy numbers of $(p,q)$-clique and near $(p,q)$-clique. Near $(p,q)$-clique is the subgraph with only one edge less than $(p,q)$-clique, as is shown in Fig. \ref{fig:near-pqclique}. The collector aggregates the noisy numbers of the two subgraphs and obtains the unbiased estimation of $(p,q)$-clique with the corrective term. The unbiasedness and utility of the process above are analyzed in Theorem \ref{theo:unbiasedness} and Theorem \ref{theo:utility}, respectively. Algorithm \ref{alg:algorithm 1} presents the overview of the baseline algorithm of traditional edge LDP.
\begin{figure}[htbp]
	\centering
	\subfloat[$(p,q)$-clique]{\label{fig:pqclique_1}\includegraphics[scale=0.4]{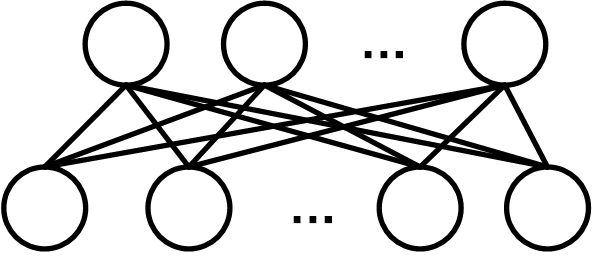}}\quad
	\subfloat[near $(p,q)$-clique]{\label{fig::near-pqclique_1}\includegraphics[scale=0.4]{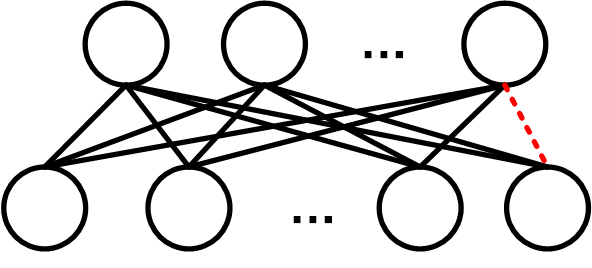}}	
	\caption{The examples of $(p,q)$-clique and near $(p,q)$-clique.}
    \label{fig:near-pqclique}
\end{figure}
\begin{algorithm}[t]
\caption{Baseline algorithm of traditional edge LDP.}
\label{alg:algorithm 1}

\KwIn{Edge neighboring list $\textbf{a}_{1},\cdots ,\textbf{a}_{n}\in \left\{ 0,1\right\}^{n}$, privacy budget $\epsilon\in \mathbb{R}$.}
\KwOut{Private unbiased estimate $\tilde{f}_{pq}^{E}(G)$ of $(p,q)$-clique.}

\BlankLine
\tcp{set the flip probability of RR.}
[$v_i$, c] $\mu\leftarrow \frac{1}{e^{\epsilon}+1}$\; 
\tcc{First round of user-collector interaction.}
\For{i from 1 to n}
{
	\tcp{$d_i$ is the degree of user $v_i$.}
	[$v_i$] $\textbf{a}_{i}^{'}\leftarrow (RR_{\mu}(a_{i,1}),\cdots ,RR_{\mu}(a_{i,d_i}))$\; 
	[$v_i$] Send $\textbf{a}_{i}^{'}$ to the collector\;
} 
\tcp{The collector obtains the noisy edge set $\textbf{E}'$ and sends it to the corresponding users.}
[c] $\textbf{E}' \leftarrow (\textbf{a}_{1}^{'},\cdots ,\textbf{a}_{n}^{'})$\;
[c] Send $\textbf{E}'$ to users\;
\tcc{Second round of user-collector interaction.}
\For{i from 1 to n}
{
	[$v_i$] receive noisy edges from the collector\;
	\tcp{CalcPQClique1 is the function of counting $(p,q)$-clique and near $(p,q)$-clique.}
	[$v_i$] $f_i,s_i \leftarrow CalcPQClique1(\textbf{a}_i,\textbf{E}')$\;
	[$v_i$] $\tilde{f}_i \leftarrow \frac{1}{(1-2\mu)^{(p-1)q}}(f_{i}-\mu s_{i})$\;
	[$v_i$] Return $\tilde{f}_i$ to the collector\;
}
[c] $\tilde{f}_{pq}^{E}(G) \leftarrow \sum_{i=1}^{n}\tilde{f}_i$\;
return $\tilde{f}_{pq}^{E}(G)$

\end{algorithm}

\subsection{K-stars LDP Algorithm}
\label{sec:k-stars LDP}
\subsubsection{Motivation}
Recall from the Section \ref{sec:Intro} that the underlying premise for edge LDP is that each user should be clearly know the existence or non-existence of their one-hop neighbors. Essentially, edge possesses such favorable properties that can be exploited in the LDP scenario: (i) edge is a simple subgraph that can be directly obfuscated; (ii) edges can form more intricate subgraph structures when combined. We notice that not only edges but also k-stars possess the above properties. For example, a user can directly obfuscate the existence of any k-stars centering on himself, since he is familiar with all his one-hop neighbors. Also, k-stars can form other complicate subgraph structures, as is illustrated in Section \ref{sec:Intro}. However, k-stars, as a basic obfuscated unit, requires far less noise than edge does, which is proved in Theorem \ref{theo:utility}. Based on this fact, we are motivated to propose the novel idea of $k$-stars LDP and the algorithm for $(p,q)$-clique enumeration. Proved by the theoretical analysis in Section \ref{sec:theoretical analysis}, $k$-stars LDP only sacrifices a relative small privacy but achieves a much better performance than traditional edge LDP.

\subsubsection{Algorithm}
We first illustrate $k$-stars neighboring list and $k$-stars sampling technique specifically. Then, we introduce the overall $k$-stars LDP algorithm with a pseudocode.\\

\noindent
\textbf{$K$-stars neighboring list.} $K$-stars neighboring list is the core concept of $k$-stars LDP, which is a further extension of the traditional edge neighboring list.
\begin{figure}[htbp]
    \centering
    \hspace{-0.1cm}
    \includegraphics[scale=0.35]{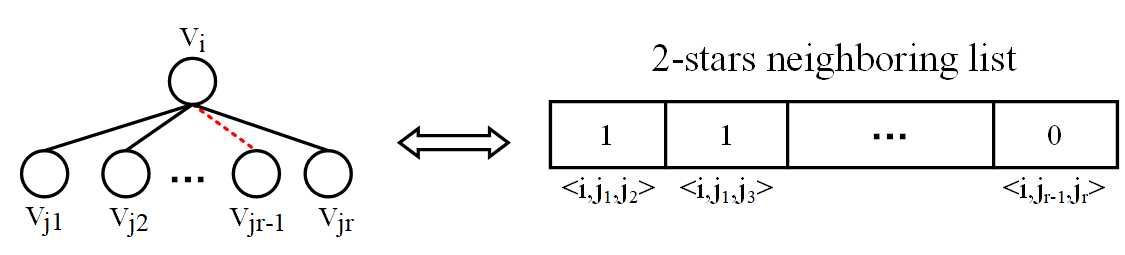}
    \caption{An example of $k$-stars neighboring list $(k=2)$.}
    \label{fig:k-stars-nl}
\end{figure}
As is shown in Fig. \ref{fig:k-stars-nl}, one position in the $k$-stars neighboring list of user $v_i$ represents the existence/inexistence of a $k$-stars centering on him, where 1 denotes the existence and 0 denotes the inexistence. $K$-stars neighboring list is a local projection of the traditional edge neighboring list since users can locally enumerate all of the $k$-stars centering on him without the knowledge of his two-hop neighbors. With the $k$-stars neighboring list, two-round user-collector interaction mechanism can be applied and the unbiased estimation of $(p,q)$-clique can be obtained similar to the baseline algorithm.\\

\noindent
\textbf{$k$-stars LDP algorithm for $(p,q)$-clique enumeration.} Algorithm \ref{alg:algorithm 2} shows the whole process of $k$-stars LDP algorithm for $(p,q)$-clique enumeration.

The first round of user-collector interaction appears from line 2 to 8 in Algorithm \ref{alg:algorithm 2}. In this round, user $v_i$ generates the $k$-stars neighboring list and obtains the noisy one with RR, sending it to the collector later. The collector receives the data and forms the noisy $k$-stars set $\textbf{KS}'$, sending it to the corresponding users later.

The second round of user-collector interaction appears from line 9 to 15. User $v_i$ receives the noisy $k$-stars set $\textbf{KS}'$ and counts the numbers of $(p,q)$-clique and $(p-1,q)$-clique locally, denoted as $f_i$ and $s_i$ respectively. Note that $f_i$ of noisy $(p,q)$-clique here is biased, thus needs to be corrected. The corrective term $s_i$ is the number of $(p-1,q)$-clique centering on user $v_i$. User $v_i$ calculates the unbiased estimate of $(p,q)$-clique through the equation $\tilde{f}_i=\frac{1}{(1-2\mu)^{p-1}}(f_i-\mu s_i)$. We will prove the correctness of this equation in Section \ref{sec:theoretical analysis}. Finally, the collector obtains the unbiased estimate of $(p,q)$-clique in the bipartite graph G.\\

\noindent
\textbf{Absolute value correction technique.} The unbiased terms of Algorithm \ref{alg:algorithm 1} and Algorithm \ref{alg:algorithm 2} can be written respectively: $\tilde{f}_i = \frac{1}{(1-2\mu)^{(p-1)q}}(f_i-\mu s_i)$ and $\tilde{f}_i = \frac{1}{(1-2\mu)^{p-1}}(f_i-\mu s_i)$, as illustrated aforementioned. When the graph is sparse and the noisy number of $(p,q)$-clique is smaller than that of near $(p,q)$-clique (or $(p-1,q)$-clique), the algorithms may generate negative values. This violates the constraints of the problem of subgraph count, since the counts must be non-negative. In order to solve this negative-output problem, we propose the technique of absolute value correction, as illustrated by equation (3) and (4):
\begin{align}
\tilde{f}_i = \frac{1}{2(1-2\mu)^{(p-1)q}}((f_i-\mu s_i)+|f_i-\mu s_i|),\tag{3}
\end{align}
\begin{align}
\tilde{f}_i = \frac{1}{2(1-2\mu)^{p-1}}((f_i-\mu s_i)+|f_i-\mu s_i|).\tag{4}
\end{align}
When the algorithms produce a negative value, it would be corrected to 0; otherwise, it would still be unbiased since the output is divided by 2. By employing this novel absolute value correction technique, we solve the negative-output problem and enhance the utility of algorithms while ensure the unbiasedness.\\

\noindent
\textbf{$k$-stars sampling technique.} Edge sampling is one of the most practical technique to improve the scalability of algorithms. In light of the superiority of edge sampling than other techniques claimed by \cite{wu2016counting}\
, the basic idea of edge sampling is borrowed in this paper. We sample the noisy $k$-stars neighboring list after applying RR on the original $k$-stars neighboring list and then propose the $k$-stars sampling technique.
 
\begin{algorithm}[t]
\caption{$K$-stars LDP algorithm for $(p,q)$-clique enumeration.}
\label{alg:algorithm 2}

\KwIn{Edge neighboring list $\textbf{a}_{1},\cdots ,\textbf{a}_{n}\in \left\{ 0,1\right\}^{n}$, privacy budget $\epsilon\in \mathbb{R}$.}
\KwOut{Private unbiased estimate $\tilde{f}_{pq}^{K}(G)$ of $(p,q)$-clique.}

\BlankLine
\tcp{set the flip probability of RR.}
[$v_i$, c] $\mu\leftarrow \frac{1}{e^{\epsilon}+1}$\; 
\tcc{First round of user-collector interaction.}
\For{i from 1 to n}
{
	\tcp{KSNLGen is the projection function from edge neighboring list to $k$-stars neighboring list.}
	[$v_i$] $\textbf{KS}_i \leftarrow KSNLGen(\textbf{a}_i)$\;
	\tcp{$t_i$ is the number of $k$-stars centering on user $v_i$.}
	[$v_i$] $\textbf{KS}_{i}^{'}\leftarrow (RR_{\mu}(KS_{i,1}),\cdots ,RR_{\mu}(KS_{i,t_i}))$\; 
	[$v_i$] Send $\textbf{KS}_{i}^{'}$ to the collector\;
} 
\tcp{The collector obtains the noisy $k$-stars set $\textbf{KS}'$ and sends it to the corresponding users.}
[c] $\textbf{KS}' \leftarrow (\textbf{KS}_{1}^{'},\cdots ,\textbf{KS}_{n}^{'})$\;
[c] Send $\textbf{KS}'$ to users\;
\tcc{Second round of user-collector interaction.}
\For{i from 1 to n}
{
	[$v_i$] receive noisy $k$-stars from the collector\;
	\tcp{CalcPQClique2 is the function of counting $(p,q)$-clique and $(p-1,q)$-clique.}
	[$v_i$] $f_i,s_i \leftarrow CalcPQClique2(\textbf{KS}_i,\textbf{KS}')$\;
	[$v_i$] $\tilde{f}_i \leftarrow \frac{1}{(1-2\mu)^{p-1}}(f_{i}-\mu s_{i})$\;
	[$v_i$] Return $\tilde{f}_i$ to the collector\;
}
[c] $\tilde{f}_{pq}^{K}(G) \leftarrow \sum_{i=1}^{n}\tilde{f}_i$\;
return $\tilde{f}_{pq}^{K}(G)$

\end{algorithm}

\subsection{Theoretical Analysis}
\label{sec:theoretical analysis}
The theoretical guarantee on the privacy, unbiasedness and utility of both edge LDP and $k$-stars LDP algorithms for $(p,q)$-clique enumeration is introduced in this section.\\

\noindent
\textbf{Privacy.}The relationship between edge LDP and $k$-stars LDP on privacy are represented in Theorem \ref{theo:privacy} below.

\begin{theorem}
\label{theo:privacy}
If a local randomized algorithm $\mathcal{R}$ provides $\epsilon$-$k$-stars LDP, it also provides $\frac{k2^{(k-1)}}{2^k-1}\epsilon$-edge LDP.
\end{theorem}

\begin{proof}
When $k=1$, 1-stars is edge. As presented in Fig. \ref{fig:privacy-one}, a flip in the 1-stars neighboring list would just cause a flip in the edge neighboring list. Thus, $\epsilon$-1-stars LDP provides $\epsilon$-edge LDP.

When $k=2$, the situation is changed because of the asymmetry of 0/1 in 2-stars neighboring list, as shown in Fig. \ref{fig:privacy-two}. When a position in 2-stars neighboring list is flipped from 0 to 1, its geometric meaning is presented as follow: there are two subgraphs, one with no edge connecting between nodes (also called three-nodes), the other with only one edge connecting two nodes out of three in total (also called one-edge-one-node). Two such subgraphs are considered as the non-existence of $2$-stars, denoted as 0 numerically. A flip from 0 to 1 is that these two subgraphs changes into a $2$-stars. However, there are 2 changes from one-edge-one-node to $2$-stars: the one with the edge on the left side changes into a $2$-stars and the one with the edge on the right side changes into a $2$-stars. Thus, there are 3 situations in total. The change from three-nodes to $2$-stars would add two edges, while the change from one-edge-one-node to $2$-stars would add one edge. Since all of the three changes are equivalent in the view of probability, a flip in the 2-stars neighboring list would cause $\frac{2\times 1 + 1\times 2}{3}=\frac{4}{3}$ flips in edge neighboring list averagely. Thus, $\epsilon$-2-stars LDP provides $\frac{4}{3}\epsilon$-edge LDP.

Similarly, the geometric meaning of a flip in the $k$-stars neighboring list is that such k subgraphs, from (k+1)-nodes to (k-1)-edges-one-node, change into a $k$-stars, as is shown in Fig. \ref{fig:privacy-k}. From the perspective of probability, a flip in the $k$-stars neighboring list would cause $\frac{\sum_{i=1}^{k} iC_{k}^{i}}{\sum_{i=1}^k C_{k}^{i}}=\frac{k2^{(k-1)}}{2^k-1}$ flips in edge neighboring list averagely. Thus, $\epsilon$-$k$-stars LDP provides $\frac{k2^{(k-1)}}{2^k-1}\epsilon$-edge LDP.

\begin{figure}[htbp]
	\hspace{-1.4cm}
    \begin{minipage}[b]{0.4\linewidth}
    		\centering
		\subfloat[][1-stars]
		{
			\includegraphics[width=0.6\linewidth]{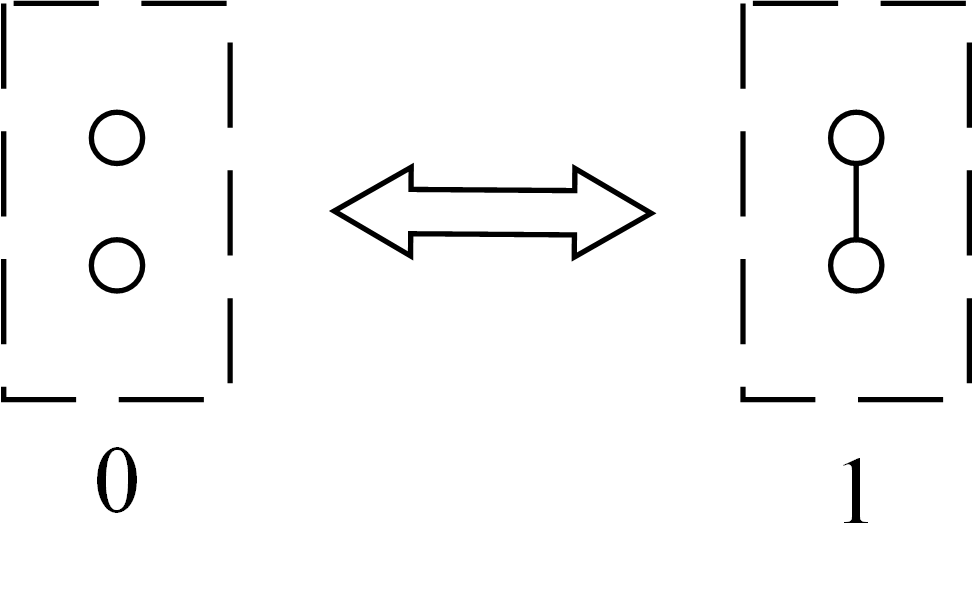}
			\label{fig:privacy-one}
		}\quad
		\subfloat[][2-stars]
		{
			\includegraphics[width=0.6\linewidth]{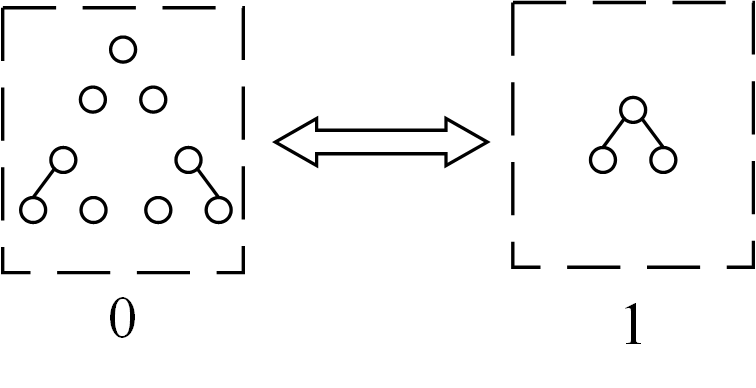}
			\label{fig:privacy-two}
		}
    \end{minipage}%\par
    \begin{minipage}[b]{0.6\linewidth}
    		\centering
		\subfloat[][$k$-stars]
		{
			\includegraphics[scale=0.25]{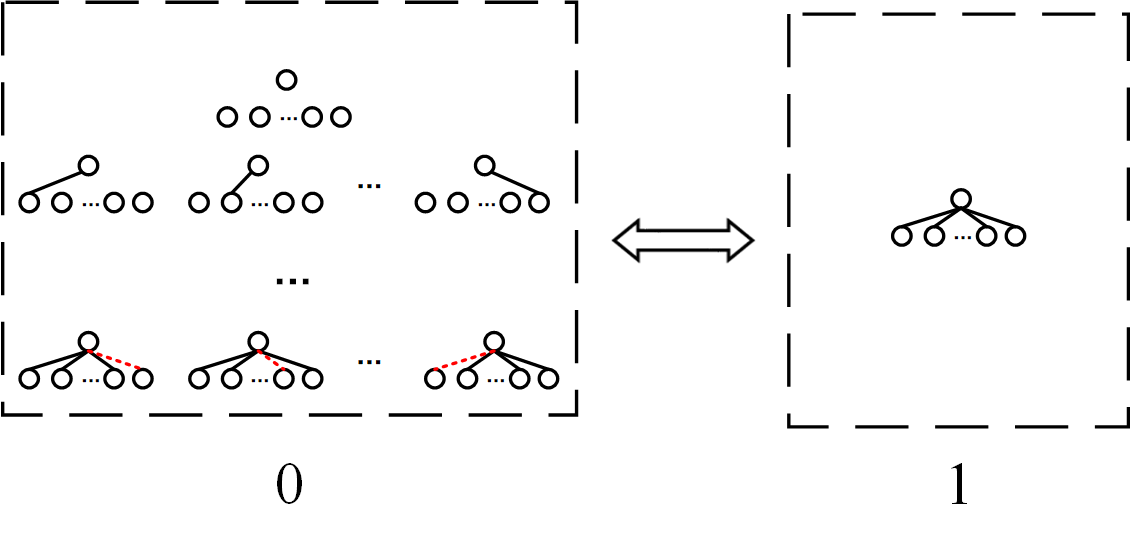}
			\label{fig:privacy-k}
		}
    \end{minipage}
    \caption{A flip in several kinds of $k$-stars neighboring list.}
    \label{fig:flip}
\end{figure}

\end{proof}

\noindent
\textbf{Unbiasedness.} The unbiasedness of a randomized algorithm is vital to the improvement of its utility, according to the theorem of bias-variance decomposition. Both algorithm \ref{alg:algorithm 1} and algorithm \ref{alg:algorithm 2} proposed in this paper are unbiased, proved by the Theorem \ref{theo:unbiasedness} below.

\begin{theorem}
\label{theo:unbiasedness}

Let $f_{pq}(G)$ be the true number of $(p,q)$-clique in graph G, $\tilde{f}_{pq}^{E}(G)$ and $\tilde{f}_{pq}^{K}(G)$ be the estimation output of algorithm \ref{alg:algorithm 1} and algorithm \ref{alg:algorithm 2} respectively. Then $E(\tilde{f}_{pq}^{E}(G))=E(\tilde{f}_{pq}^{K}(G))= f_{pq}(G)$.

\end{theorem}

\begin{proof}
Let $\tilde{f}_{pq}^{E}(G)$ be the output of Algorithm \ref{alg:algorithm 1}, $f_i$ and $s_i$ be the noisy count of $(p,q)$-clique and near $(p,q)$-clique, $i=i_1$ be the index of user $v_i$. We assume that $1\le j_1\le j_2\le \cdots \le j_q\le n$ be the indexes of user $v_{j_{1}}$ to user $v_{j_{q}}$. For $f_i$ and $s_i$, we have:
$$
\begin{aligned}
\mathbb{E}(f_i)&=\sum_{1\le j_1\le j_2\le \cdots \le j_q\le n} \mathbb{E}(a_{i_{1}j_{1}}\cdots a_{i_{1}j_{q}}a'_{i_{2}j_{1}}\cdots a'_{i_{p}j_{q}})\\
&=\sum_{a_{i_{1}j_{1}}=\cdots =a_{i_{1}j_{q}}=1} \mathbb{E}(\prod_{2\le l\le p,1\le r\le q} a'_{i_{l}j_{r}})\\
\end{aligned}
$$
$$
\begin{aligned}
\mathbb{E}(s_i)&=\sum_{1\le j_1\le j_2\le \cdots \le j_q\le n} \mathbb{E}(a_{i_{1}j_{1}}\cdots a_{i_{1}j_{q}}a'_{i_{2}j_{1}}\cdots a'_{i_{p}j_{q-1}})\\
&=\sum_{a_{i_{1}j_{1}}=\cdots =a_{i_{1}j_{q}}=1} \mathbb{E}(\frac{\prod_{2\le l\le p,1\le r\le q} a'_{i_{l}j_{r}}}{a^{'}_{i_{p}j_{q}}})\\
\end{aligned}
$$
From RR, we know that:
$$
\begin{aligned}
\textbf{a}' & =(1-\mu )\textbf{a}+\mu (\textbf{1}-\textbf{a})\\
& =(1-2\mu)\textbf{a}+\mu\cdot \textbf{1}
\end{aligned}
$$
Thus:
$$
\begin{aligned}
\mathbb{E}(f_i)&=\sum_{a_{i_{1}j_{1}}=\cdots =a_{i_{1}j_{q}}=1} \prod_{2\le l \le p,1\le r\le q}((1-2\mu)a_{i_{l}j_{r}}+\mu)\\
&=\sum_{a_{i_{1}j_{1}}=\cdots =a_{i_{1}j_{q}}=1} ((1-2\mu)^{(p-1)q}\prod_{2\le l\le p,1\le r\le q}a_{i_{l}j_{r}}\\
		&+(1-2\mu)^{(p-1)q-1}\mu(\sum_{2\le l_{1}\le p, 1\le r_{1}\le q}\frac{\prod_{2\le l\le p,1\le r\le q}a_{i_{l}j_{r}}}{a_{l_{1}r_{1}}})\\
		&+\cdots +(1-2\mu)\mu^{(p-1)q-1}(\sum_{2\le l\le p,1\le r\le q}a_{i_{l}j_{r}})+\mu^{(p-1)q})\\
\end{aligned}
$$
$$
\begin{aligned}
\mathbb{E}(s_i)&=\sum_{a_{i_{1}j_{1}}=\cdots =a_{i_{1}j_{q}}=1} \frac{\prod_{2\le l \le p,1\le r\le q}((1-2\mu)a_{i_{l}j_{r}}+\mu)}{(1-2\mu)a_{i_{p}j_{q}}+\mu}\\
&=\sum_{a_{i_{1}j_{1}}=\cdots =a_{i_{1}j_{q}}=1} ((1-2\mu)^{(p-1)q-1}\prod_{2\le l\le p,1\le r\le q}a_{i_{l}j_{r}}\\
		&+(1-2\mu)^{(p-1)q-2}\mu(\sum_{2\le l_{1}\le p, 1\le r_{1}\le q}\frac{\prod_{2\le l\le p,1\le r\le q}a_{i_{l}j_{r}}}{a_{l_{1}r_{1}}})\\
		&+\cdots +(1-2\mu)\mu^{(p-1)q-2}(\sum_{2\le l\le p,1\le r\le q}a_{i_{l}j_{r}})+\mu^{(p-1)q-1})\\
\end{aligned}
$$
We notice that $\frac{1}{(1-2\mu)^{(p-1)q}}(f_i-\mu s_i)=\sum_{a_{i_{1}j_{1}}=\cdots =a_{i_{1}j_{q}}=1}\\
\prod_{2\le l\le p,1\le r\le q}a_{i_{l}j_{r}}$, which is the number of $(p,q)$-clique centering on user $v_i$. Thus, we have:
$$
\begin{aligned}
\mathbb{E}(\tilde{f}_{pq}^{E}(G))&=\sum_{i=1}^n \mathbb{E}(\tilde{f}_i)\\
&=\sum_{i=1}^n \frac{1}{(1-2\mu)^{(p-1)q}}\mathbb{E}(f_i-\mu s_i)\\
&=\sum_{i=1}^n \sum_{a_{i_{1}j_{1}}=\cdots =a_{i_{1}j_{q}}=1} \prod_{2\le l\le p,1\le r\le q}a_{i_{l}j_{r}}\\
&=f_{pq}(G)\\
\end{aligned}
$$
Thus, the unbiasedness of Algorithm \ref{alg:algorithm 1} is proved. 

Similarly, let $\tilde{f}_{pq}^{K}(G)$ be the output of Algorithm \ref{alg:algorithm 2}, $f_i$ and $s_i$ be the noisy count of $(p,q)$-clique and $(p-1,q)$-clique, $i=i_1$ be the index of user $v_i$. For $1\le l\le p$, $\textbf{KS}_{l}$ represents the $l$-th $q$-stars neighboring list. For $f_i$ and $s_i$, we have:
$$
\begin{aligned}
\mathbb{E}(f_i)&=\sum_{\textbf{KS}_{1}=1} ((1-2\mu)^{p-1}\prod_{2\le l\le p}\textbf{KS}_{l}\\
&+(1-2\mu)^{p-2}\mu(\sum_{2\le l_{1}\le p}\frac{\prod_{2\le l\le p}\textbf{KS}_{l}}{\textbf{KS}_{l_{1}}})\\
		&+\cdots +(1-2\mu)\mu^{p-2}(\sum_{2\le l\le p}\textbf{KS}_{l})+\mu^{p-1})\\
\end{aligned}
$$
$$
\begin{aligned}
\mathbb{E}(s_i)&=\sum_{\textbf{KS}_{1}=1} ((1-2\mu)^{p-2}\prod_{2\le l\le p-1}\textbf{KS}_{l}\\
&+(1-2\mu)^{p-3}\mu(\sum_{2\le l_{1}\le p-1}\frac{\prod_{2\le l\le p-1}\textbf{KS}_{l}}{\textbf{KS}_{l_{1}}})\\
		&+\cdots +(1-2\mu)\mu^{p-3}(\sum_{2\le l\le p-1}\textbf{KS}_{l})+\mu^{p-2})\\
\end{aligned}
$$
Namely, $\frac{1}{(1-2\mu)^{p-1}}(f_i-\mu s_i)=\sum_{\textbf{KS}_{1}=1} \prod_{2\le l\le p}\textbf{KS}_{l}$. Thus:
$$
\begin{aligned}
\mathbb{E}(\tilde{f}_{pq}^{K}(G))&=\sum_{i=1}^n \mathbb{E}(\tilde{f}_i)\\
&=\sum_{i=1}^n \frac{1}{(1-2\mu)^{p-1}}\mathbb{E}(f_i-\mu s_i)\\
&=\sum_{i=1}^n \sum_{\textbf{KS}_{1}=1} \prod_{2\le l\le p}\textbf{KS}_{l}\\
&=f_{pq}(G)\\
\end{aligned}
$$
As a result, Algorithm \ref{alg:algorithm 2} is unbiased.
\end{proof}

\noindent
\textbf{Utility.} The utility of a randomized algorithm, measured by the expected $L_{2}$ loss $L_{2}^{2}(f_{pq}(G),\tilde{f}_{pq}(G))$, is composed of the bias and variance, as per the bias-variance decomposition theorem. Since both algorithm \ref{alg:algorithm 1} and algorithm \ref{alg:algorithm 2} are proved unbiased by Theorem \ref{theo:unbiasedness}, the expected $L_{2}$ loss is equivalent to the upper-bound of two algorithms’ variance. This is given by the following lemma and theorem:

\begin{lemma}
\label{lemma:lem}
Let $c_i=\prod_{1\le r\le q}a_{ij_{r}}$, $N_{2,q}(G),N_{1,q}(G)$ be the numbers of (2,q)-clique and (1,q)-clique in graph G, respectively. Then,
$$
\begin{aligned}
\sum_{1\le i,j_1,j_2,\cdots ,j_q\le n} c_{i}^2 \le 2N_{2,q}(G)+N_{1,q}(G)\\
\end{aligned}
$$
\end{lemma}

\begin{theorem}
\label{theo:utility}

Let $\epsilon\in \mathbb{R}$, $\mu=\frac{1}{e^{\epsilon}+1}$ be the flip probability in RR, $f_{pq}(G)$ be the true number of $(p,q)$-clique in graph G, $\tilde{f}_{pq}^{E}(G)$ and $\tilde{f}_{pq}^{K}(G)$ be the estimation output of algorithm \ref{alg:algorithm 1} and algorithm \ref{alg:algorithm 2}, respectively. Let $S=2N_{2,q}(G)+N_{1,q}(G)$ be the coefficient, where $N_{2,q}(G)$ and $N_{1,q}(G)$ are the numbers of $(2,q)$-clique and $(1,q)$-clique in graph G, respectively. Then:
\begin{align}
V(\tilde{f}_{pq}^{E}(G)) &\le \frac{\mu(1-\mu)}{(1-2\mu)^{2(p-1)q}}((p-1)q-((p-1)q+1)\mu^{2})S,\tag{5}
\end{align}
\begin{align}
V(\tilde{f}_{pq}^{K}(G)) \le \frac{\mu(1-\mu)}{(1-2\mu)^{2(p-1)}}((p-1)-p\mu^{2})S.\tag{6}
\end{align}
\end{theorem}

\begin{proof}
We start with the proof of the lemma:\\
\noindent
\textbf{The proof of the lemma \ref{lemma:lem}.} We notice that $\sum_{1\le i\le n} c_{i}^{2}=\sum_{1\le i,j_1,j_2,\cdots ,j_q\le n} (c_i+c_i(c_i-1))$. The geometric meaning of left side is the number of (1,q)-clique in graph G. The right side of the equation, denoted as $\sum_{1\le i,j_1,j_2,\cdots ,j_q\le n}\\ c_i(c_i-1)=\sum_{1\le i,j_1,j_2,\cdots ,j_q\le n} 2\tbinom{c_i}{2}$, is the twice of the number of (2,q)-clique in graph G, since the number of (2,q)-clique is the  number of randomly and non-repetitively selecting two (1,q)-clique. Thus, the lemma is proved. 

The proof of the upper-bound of the variance of Algorithm \ref{alg:algorithm 1} is similar to that of Algorithm \ref{alg:algorithm 2}. We omit the proof of Algorithm \ref{alg:algorithm 1} for simplicity.

\noindent
\textbf{The upper-bound of the variance of Algorithm 2.} From the unbiased correction of Algorithm 2, we have:
$$
\begin{aligned}
\mathbb{V}(\tilde{f}_i)&=\mathbb{V}(\frac{1}{(1-2\mu)^{p-1}}(f_i-\mu s_i))\\
&=\frac{1}{(1-2\mu)^{2(p-1)}}(\mathbb{V}(f_i)+\mu^{2}\mathbb{V}(s_i)-2\mu Cov(f_i,s_i))\\
\end{aligned}
$$
Let $1\le j_1,j_2,\cdots ,j_q\le n$ be the indexes of q arbitrary users. From the left side of the variance, we have:
$$
\begin{aligned}
\mathbb{V}(f_i)&=\mathbb{V}(\sum_{1\le j_1,j_2,\cdots ,j_q\le n}(\textbf{KS}_1)\cdot (\prod_{2\le l\le p}\textbf{KS}'_{l}))\\
&=\sum_{1\le j_1,j_2,\cdots ,j_q\le n}(\prod_{1\le r\le q}a_{ij_{r}})^{2}\mathbb{V}(\prod_{2\le l\le p}\textbf{KS}'_{l})\\
&=\sum_{1\le j_1,j_2,\cdots ,j_q\le n}c_{i}^{2}\mathbb{V}(\prod_{2\le l\le p}\textbf{KS}'_{l})\\
\end{aligned}
$$
We notice that $\prod_{2\le l\le p}\textbf{KS}'_{l}$ is a (p-1)-Bernoulli random variable. Thus, $\mathbb{V}(\prod_{2\le l\le p}\textbf{KS}'_{l})=(p-1)\mu(1-\mu)$. Then, with the Lemma \ref{lemma:lem}, we can obtain the inequality below:
$$
\begin{aligned}
\sum_{1\le i\le n} \mathbb{V}(f_i)\le (p-1)\mu(1-\mu)(2N_{2,q}(G)+N_{1,q}(G))\\
\end{aligned}
$$
Similarly,
$$
\begin{aligned}
\sum_{1\le i\le n} \mu^{2}\mathbb{V}(s_i)\le ((p-2)\mu^3(1-\mu)(2N_{2,q}(G)+N_{1,q}(G))\\
\end{aligned}
$$
As for the right side of the variance, we have:
$$
\begin{aligned}
2\mu Cov(f_i,s_i)&=2\mu \mathbb{E}((f_i-\mathbb{E}(f_i))(s_i-\mathbb{E}(s_i)))\\
\end{aligned}
$$
We notice that:
$$
\begin{aligned}
f_i&=\textbf{KS}^{'}_{l}s_i\\
\mathbb{E}(f_i)&=\mathbb{E}(\textbf{KS}^{'}_{l}s_i)\\
&=\mu \mathbb{E}(s_i)\\
\end{aligned}
$$
Then, we have:
$$
\begin{aligned}
2\mu Cov(f_i,s_i)&=2\mu \mathbb{E}((\textbf{KS}^{'}_{l}s_i-\mu \mathbb{E}(s_i))(s_i-\mathbb{E}(s_i)))\\
&=2\mu^2 \mathbb{E}(s_i-\mathbb{E}(s_i))^2\\
&=2\mu^2 \mathbb{V}(s_i)\\
\end{aligned}
$$
With the equation of $\mathbb{V}(s_i)$ above, we have:
$$
\begin{aligned}
\sum_{1\le i\le n} 2\mu Cov(f_i,s_i)&\le 2(p-2)\mu^3(1-\mu)(2N_{2,q}(G)+N_{1,q}(G))\\
\end{aligned}
$$
Combining the upper-bound of $\sum_{1\le i\le n}\mathbb{V}(f_i),\sum_{1\le i\le n}\mathbb{V}(s_i)$ and $\sum_{1\le i\le n}2\mu Cov(f_i,s_i)$ together, we have:
$$
\begin{aligned}
\mathbb{V}(\tilde{f}_{pq}^{K}(G))&=\mathbb{V}(\sum_{i=1}^{n} \tilde{f}_i)\\
&=\sum_{i=1}^{n} \mathbb{V}(\tilde{f}_i)\\
&=\frac{1}{(1-2\mu)^{2(p-1)}} \sum_{i=1}^{n}(\mathbb{V}(f_i)+\mu^{2}\mathbb{V}(s_i)-2\mu Cov(f_i,s_i))\\
&\le \frac{\mu(1-\mu)}{(1-2\mu)^{2(p-1)}} ((p-1)-p\mu^2)(2N_{2,q}(G)+N_{1,q}(G))\\
\end{aligned}
$$
Thus, Theorem \ref{theo:utility} is proved.
\end{proof}

It can be seen that the left coefficient of the upper-bound of two algorithms’ variance dominates the change. When $k=q$, algorithm \ref{alg:algorithm 2} can improve the utility by $\frac{1}{(1-2\mu)^{2(p-1)(q-1)}}$, while it only sacrifices a relative small privacy by $\frac{q-1}{2}\epsilon$.\\

\noindent
\textbf{The highlights of our $k$-stars LDP algorithm.} Our proposed $k$-stars LDP algorithm has a better utility than traditional edge LDP algorithm. Moreover, when the target subgraph becomes more and more complex (i.e. p, q become larger and larger), it becomes more and more efficient to trade privacy for utility, which means that the proposed $k$-stars LDP algorithm can handle complex subgraph enumeration problems better than traditional edge LDP algorithm.

\section{Experiments}
\label{sec:exp}
In this section, we first describe the datasets used in the experiments. Then, the details of experiments are introduced, as well as the analysis.

\subsection{Datasets}	
The algorithms in this paper are evaluated on four real datasets: Gplus, IMDB, GitHub and Facebook, as introduced below. 

The Gplus dataset (namely Google+ dataset) is a dataset constructed by Google, with the information collected from chrome users, to analyze users’ social network. In this dataset, a node represents a user and an edge represents the social relationship between two users. Then, a bipartite graph $G=(V=(U,L),E)$ can be constructed with 107,614 nodes and 12,238,285 edges, as is shown in Table \ref{tab:dataset}.

The IMDB dataset (namely Internet Movie Database) is a dataset about the relationship between movies and their actors. In this dataset, a node represents an actor/actress, and an edge represents a movie played by these two actors/actresses together. Then, a bipartite graph $G=(V=(U,L),E)$ can be constructed with 896,308 nodes and 57,064,385 edges, as is shown in Table \ref{tab:dataset}.

The GitHub dataset is a dataset about the relationship between the GitHub users and the repositories they are watching. In this dataset, there are two kinds of nodes: user and repository. The edges between these two kinds of nodes represent the relationship of users watching repositories. A bipartite graph $G=(V=(U,L),E)$ can be constructed with 177,386 nodes and 440,237 edges, as is shown in Table \ref{tab:dataset}.

The Facebook dataset is a dataset about the social network of users using Facebook. In this dataset, an node represents an user and an edge represents the friendship between users. A bipartite graph $G=(V=(U,L),E)$ can be constructed with 63,732 nodes and 1,545,686 edges, as is shown in Table \ref{tab:dataset}.
\begin{table}[h!t]
\caption{The datasets employed in experiments.}
\centering
\begin{tabular}{@{}ccc@{}}
\toprule
Dataset & Node  & Edge \\ \midrule
Gplus   & 107,614  & 12,238,285   \\
IMDB   & 896,308 & 57,064,385   \\
GitHub   & 177,386 & 440,237   \\
Facebook   & 63,732 & 1,545,686   \\ \bottomrule
\end{tabular}
\label{tab:dataset}
\end{table}
\subsection{Experimental Results}
\textbf{Performance comparison.} We first compare the performance of our proposed $k$-stars LDP algorithm for $(p,q)$-clique enumeration to that of the traditional edge LDP algorithm. Fig. \ref{fig:performance} shows the relative errors of the two algorithms when $n=10^4$, $\epsilon=0.1$, $k=p=q=2$.
\begin{figure}[htbp]
	\centering
	 \includegraphics[scale=0.2]{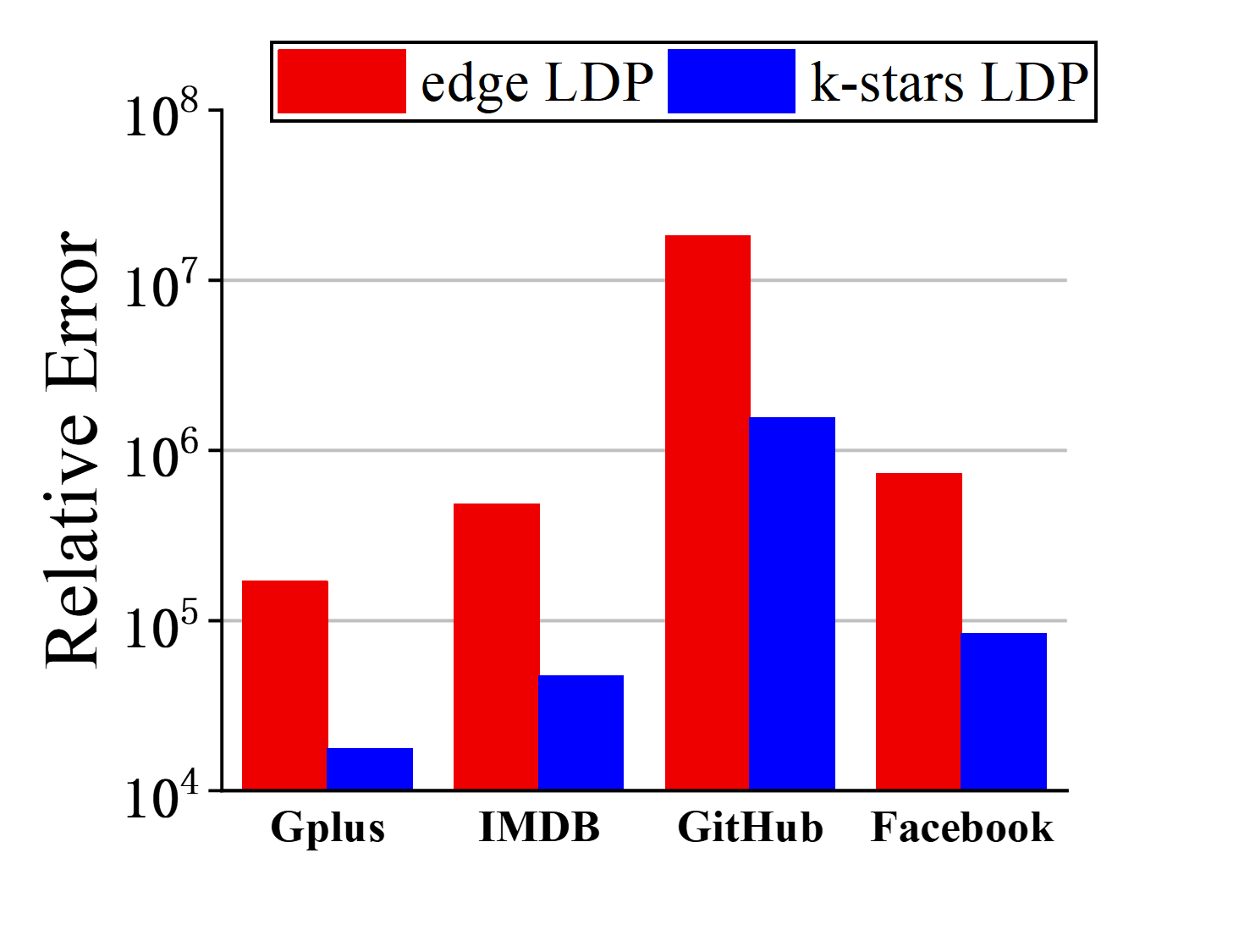}
	\caption{Relative error of edge LDP and $k$-stars LDP for $(p,q)$-clique enumeration ($n=10^4,\epsilon =0.1,k=p=q=2$).}
    \label{fig:performance}
\end{figure}
It can be seen from Fig. \ref{fig:performance} that our $k$-stars LDP algorithm for $(p,q)$-clique enumeration outperforms the traditional edge LDP algorithm. This illustrates the effectiveness of $k$-stars LDP on the optimization of utility. With the $2$-stars neighboring list, users needs only to add noise once instead of twice, compared to traditional edge LDP. The required noise of users is reduced, and thus the accuracy performance is greatly improved. Meanwhile, as per Theorem \ref{theo:privacy}, $\epsilon$-2-stars LDP algorithm also provides $\frac{4}{3}\epsilon$-edge LDP, which means it only sacrifices a relative small privacy budget compared to the traditional edge LDP algorithm.\\

\noindent
\textbf{Changing the number of nodes.} The performances of our $k$-stars LDP algorithm and traditional edge LDP algorithm are evaluated when the number of users is changed. We randomly select certain amount of users and build subgraphs from the original graph G in Gplus, IMDB, GitHub and Facebook. Then, the relative errors of the two algorithms are evaluated on these generated subgraphs. We set $\epsilon=0.1$, $k=p=q=2$, and the results of relative error and $L_2$ loss are shown in Fig. \ref{fig:re_changing-n} and Fig. \ref{fig:l2_changing-n}.
\begin{figure}[htbp]
	\centering
	\includegraphics[scale=0.37]{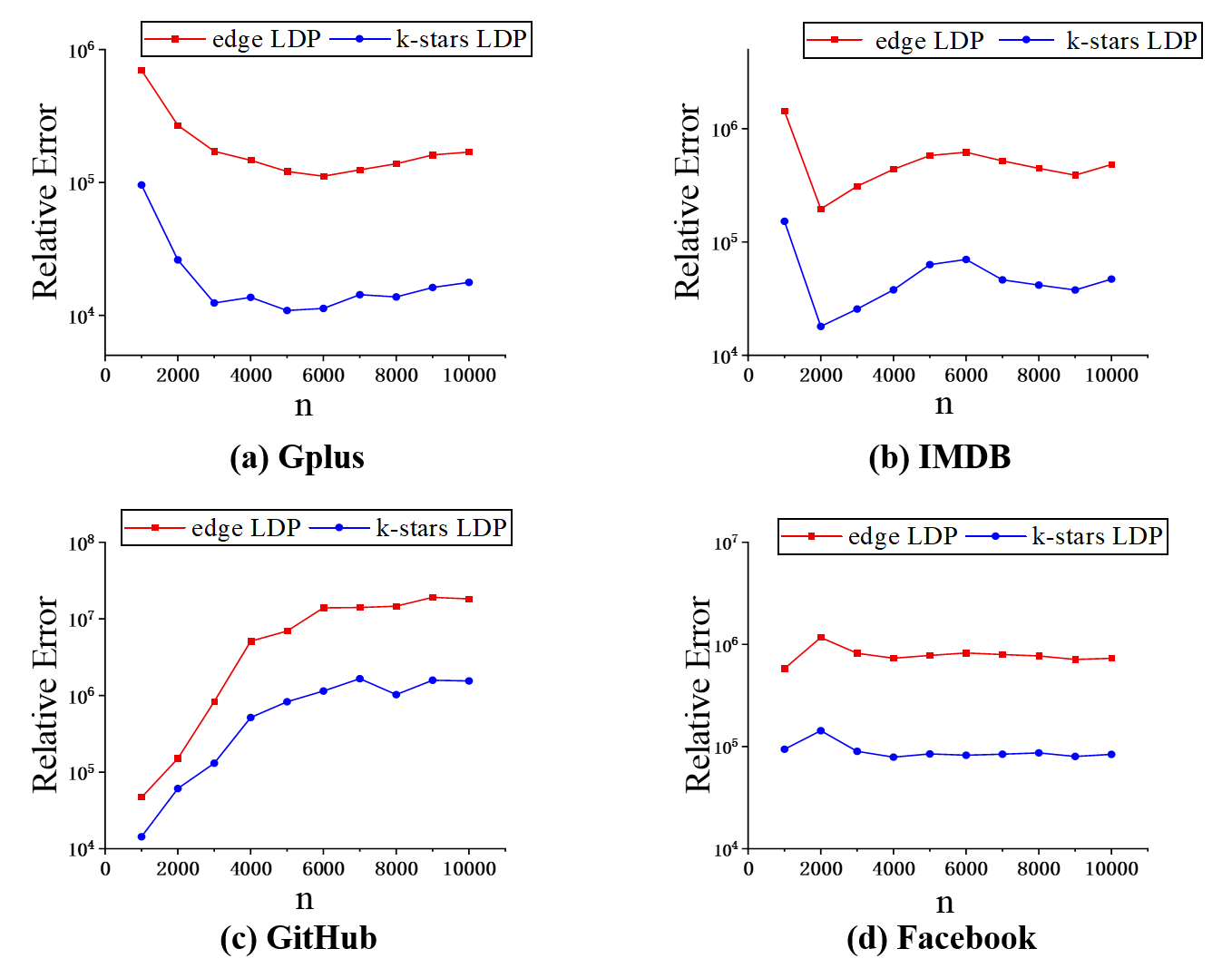}
	\caption{Relative error of edge LDP and $k$-stars LDP with various n ($\epsilon =0.1,k=p=q=2$).}
    \label{fig:re_changing-n}
\end{figure}
\begin{figure}[htbp]
	\centering
	\includegraphics[scale=0.37]{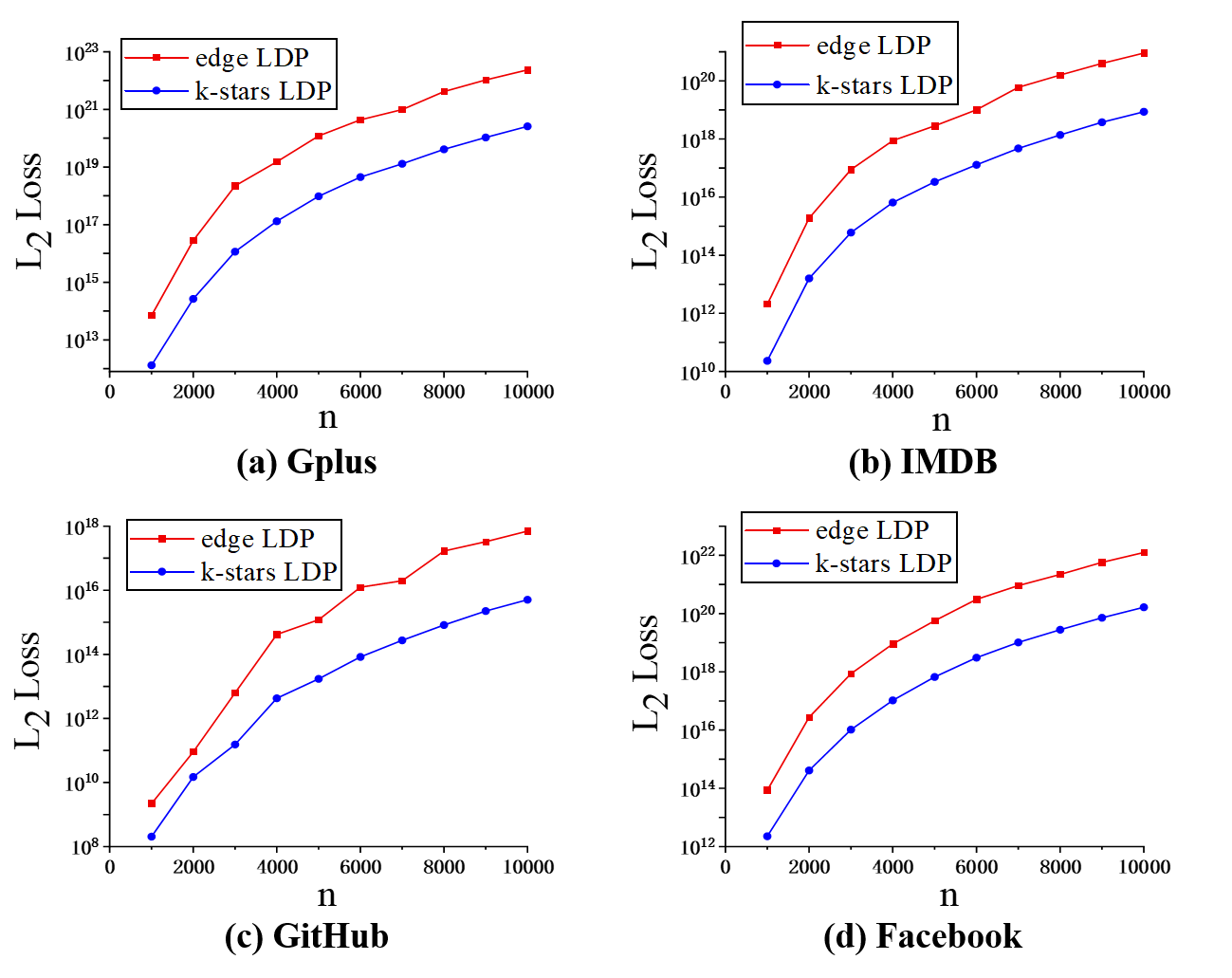}	
	\caption{$L_2$ loss of edge LDP and $k$-stars LDP with various n ($\epsilon =0.1,k=p=q=2$).}
    \label{fig:l2_changing-n}
\end{figure}
It can be seen from Fig. \ref{fig:re_changing-n} and Fig. \ref{fig:l2_changing-n} that $k$-stars LDP algorithm gains a better utility than edge LDP algorithm in all instances, as n changes. In Gplus and Facebook, the $L_2$ loss of both edge LDP algorithm and $k$-stars LDP algorithm shows an incline trend as n increases, and the relative error shows a decline trend. This is because the computational complexity of $L_2$ loss can be expressed as $O(f_{22}(G)+f_{12}(G))$, where $f_{22}(G)$ and $f_{12}(G)$ are the number of $(2,2)$-clique and $(1,2)$-clique in G respectively, according to Theorem \ref{theo:utility}. The number of the two subgraphs grows as n increases. Thus, the $L_2$ loss increases as n increases. Also, the computational complexity of relative error is $O(\sqrt{\frac{f_{12}(G)}{f_{22}^{2}(G)}})$. Since graphs of Gplus and Facebook are dense, the number of $(2,2)$-clique is larger than that of $(1,2)$-clique. As a result, the relative error declines as n increases.

However, the situation is changed in IMDB and GitHub. This is because the graphs of IMDB and GitHub are sparser than those of Gplus and Facebook. The computational complexity of $L_2$ loss and relative error are $O(f_{22}(G)+f_{12}(G))$ and $O(\sqrt{\frac{f_{12}(G)}{f_{22}^{2}(G)}})$, as illustrated above. Because of the sparsity of the graphs in IMDB and GitHub, the number of $(2,2)$-clique is lower than that of $(1,2)$-clique. Thus, the trends of IMDB and GitHub are different from those of Gplus and Facebook. The numbers of $(2,2)$-clique and $(1,2)$-clique in Gplus, IMDB, GitHub and Facebook are listed in Table \ref{tab:GI_2N+1N} and Table \ref{tab:GF_2N+1N}.

We notice that the $L_2$ loss of $k$-stars LDP in GitHub is smaller than that of the rest datasets, as shown in Fig. \ref{fig:l2_changing-n}. It can be concluded that our $k$-stars LDP algorithms work better in a sparse graph, which means it can work better in the practical scenario as real social network tends to be sparse.\\
\begin{table*}[hpt]
\caption{The number of $(2,2)$-clique and $(1,2)$-clique in Gplus and IMDB.}
\centering
\begin{tabular}{ccccccc}
\toprule
\multirow{2}{*}{Node} 
& \multicolumn{3}{c}{Gplus}  & \multicolumn{3}{c}{IMDB}  \\                                                                     
\cline{2-7} & \multicolumn{1}{c}{$(2,2)$-clique} & \multicolumn{1}{c}{$(1,2)$-clique} & $2f_{22}(G)+f_{12}(G)$ & \multicolumn{1}{c}{$(2,2)$-clique} & \multicolumn{1}{c}{$(1,2)$-clique} & $2f_{22}(G)+f_{12}(G)$ \\ 
\midrule
1000  & \multicolumn{1}{c}{12}           & \multicolumn{1}{c}{214}          & 238                    & \multicolumn{1}{c}{0}            & \multicolumn{1}{c}{32}           & 32                     \\ 
2000  & \multicolumn{1}{c}{624}          & \multicolumn{1}{c}{1788}         & 3036                   & \multicolumn{1}{c}{222}          & \multicolumn{1}{c}{451}          & 895                    \\ 
3000  & \multicolumn{1}{c}{8698}         & \multicolumn{1}{c}{9284}         & 26680                  & \multicolumn{1}{c}{958}          & \multicolumn{1}{c}{1960}         & 3876                   \\ 
4000  & \multicolumn{1}{c}{26604}        & \multicolumn{1}{c}{20364}        & 73572                  & \multicolumn{1}{c}{2136}         & \multicolumn{1}{c}{4722}         & 8994                   \\ 
5000  & \multicolumn{1}{c}{90542}        & \multicolumn{1}{c}{44973}        & 226057                 & \multicolumn{1}{c}{2900}         & \multicolumn{1}{c}{7577}         & 13377                  \\ 
6000  & \multicolumn{1}{c}{187412}       & \multicolumn{1}{c}{75553}        & 450377                 & \multicolumn{1}{c}{5132}         & \multicolumn{1}{c}{12193}        & 22457                  \\ 
7000  & \multicolumn{1}{c}{251920}       & \multicolumn{1}{c}{103733}       & 607573                 & \multicolumn{1}{c}{14858}        & \multicolumn{1}{c}{22574}        & 52290                  \\ 
8000  & \multicolumn{1}{c}{466972}       & \multicolumn{1}{c}{173737}       & 1107681                & \multicolumn{1}{c}{28188}        & \multicolumn{1}{c}{32898}        & 89274                  \\ 
9000  & \multicolumn{1}{c}{318545}       & \multicolumn{1}{c}{246808}       & 883898                 & \multicolumn{1}{c}{25723}        & \multicolumn{1}{c}{47266}        & 98712                  \\ 
10000 & \multicolumn{1}{c}{453729}       & \multicolumn{1}{c}{339028}       & 1246486                & \multicolumn{1}{c}{31100}        & \multicolumn{1}{c}{64122}        & 126322                 \\ 
\bottomrule
\end{tabular}
\label{tab:GI_2N+1N}
\end{table*}
\begin{table*}[hpt]
\caption{The number of $(2,2)$-clique and $(1,2)$-clique in GitHub and Facebook.}
\centering
\begin{tabular}{ccccccc}
\toprule
\multirow{2}{*}{Node}
& \multicolumn{3}{c}{GitHub}   & \multicolumn{3}{c}{Facebook}         \\ 
\cline{2-7} & \multicolumn{1}{c}{$(2,2)$-clique} & \multicolumn{1}{c}{$(1,2)$-clique} & $2f_{22}(G)+f_{12}(G)$ & \multicolumn{1}{c}{$(2,2)$-clique} & \multicolumn{1}{c}{$(1,2)$-clique} & $2f_{22}(G)+f_{12}(G)$ \\ 
\midrule
1000  & \multicolumn{1}{c}{0}            & \multicolumn{1}{c}{2}            & 2                      & \multicolumn{1}{c}{16}           & \multicolumn{1}{c}{262}          & 294                    \\
2000  & \multicolumn{1}{c}{0}            & \multicolumn{1}{c}{8}            & 8                      & \multicolumn{1}{c}{142}          & \multicolumn{1}{c}{1989}         & 2273                   \\
3000  & \multicolumn{1}{c}{0}            & \multicolumn{1}{c}{23}           & 23                     & \multicolumn{1}{c}{1132}         & \multicolumn{1}{c}{6993}         & 9257                   \\
4000  & \multicolumn{1}{c}{0}            & \multicolumn{1}{c}{111}          & 111                    & \multicolumn{1}{c}{4110}         & \multicolumn{1}{c}{16869}        & 25089                  \\
5000  & \multicolumn{1}{c}{0}            & \multicolumn{1}{c}{164}          & 164                    & \multicolumn{1}{c}{9656}         & \multicolumn{1}{c}{33895}        & 53207                  \\
6000  & \multicolumn{1}{c}{8}            & \multicolumn{1}{c}{366}          & 382                    & \multicolumn{1}{c}{21274}        & \multicolumn{1}{c}{63022}        & 105570                 \\
7000  & \multicolumn{1}{c}{10}           & \multicolumn{1}{c}{472}          & 492                    & \multicolumn{1}{c}{38016}        & \multicolumn{1}{c}{95924}        & 171956                 \\
8000  & \multicolumn{1}{c}{28}           & \multicolumn{1}{c}{945}          & 1001                   & \multicolumn{1}{c}{60972}        & \multicolumn{1}{c}{134591}       & 256535                 \\
9000  & \multicolumn{1}{c}{30}           & \multicolumn{1}{c}{1259}         & 1319                   & \multicolumn{1}{c}{105914}       & \multicolumn{1}{c}{191980}       & 403808                 \\
10000 & \multicolumn{1}{c}{46}           & \multicolumn{1}{c}{1676}         & 1768                   & \multicolumn{1}{c}{152972}       & \multicolumn{1}{c}{258821}       & 564765                 \\
\bottomrule
\end{tabular}
\label{tab:GF_2N+1N}
\end{table*}

\noindent
\textbf{Evaluating different $(p,q)$-cliques.} We also evaluate the performances of edge LDP algorithm and $k$-stars LDP algorithm for different $(p,q)$-cliques. In the experiment, we set $n=10^3$, $\epsilon = 0.1$, $k=2$ for $(2,2)$-clique and $k=3$ for $(2,3)$-clique and $(3,3)$-clique. Fig. \ref{fig:different-pq} shows the results of the two algorithms’ performances for different $(p,q)$-cliques.
\begin{figure}[htbp]
	\centering
	\includegraphics[scale=0.37]{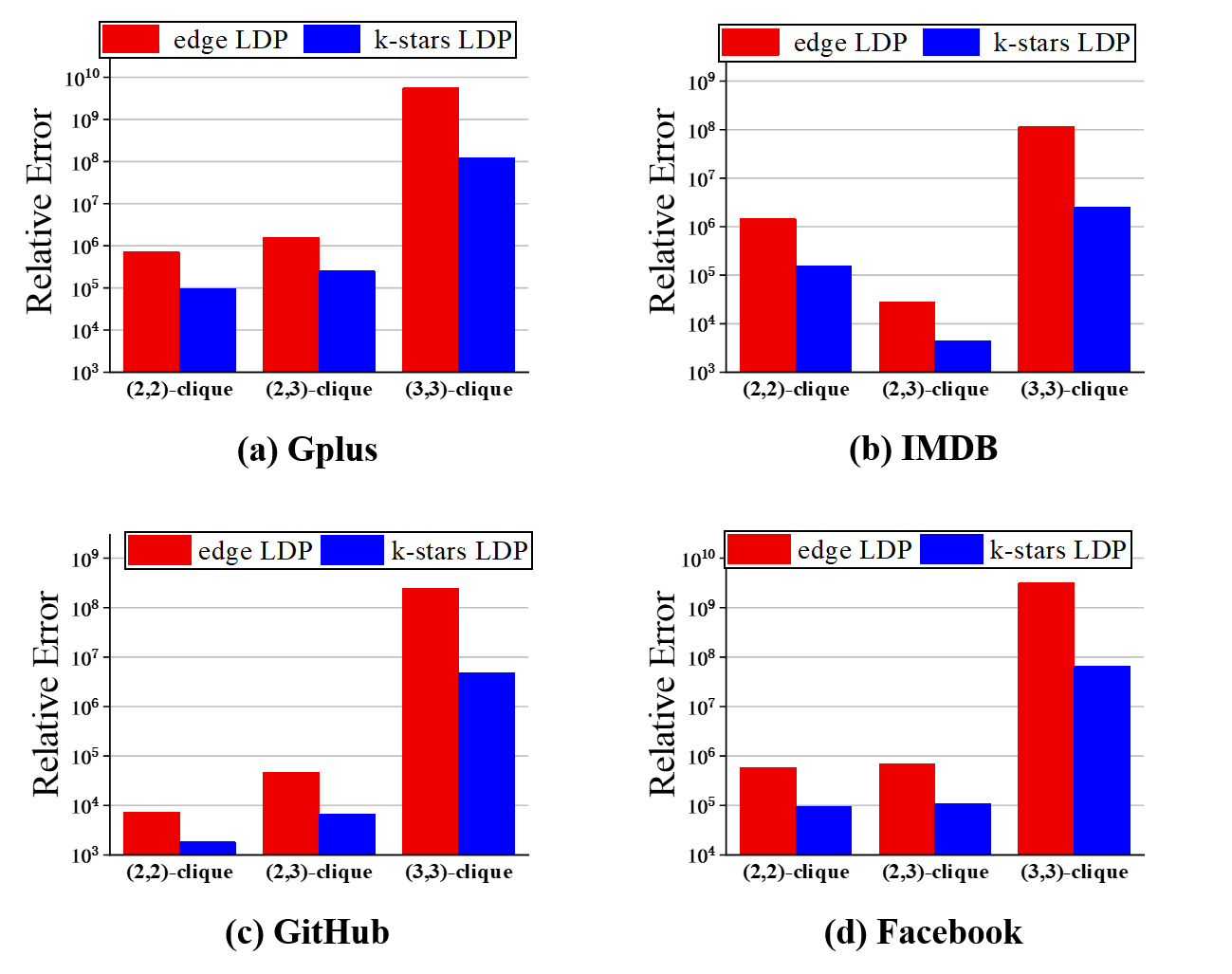}
	\caption{Relative error of edge LDP and $k$-stars LDP for different $(p,q)$-cliques ($n=10^3, \epsilon =0.1$).}
    \label{fig:different-pq}
\end{figure}
It can be seen from Fig. \ref{fig:different-pq} that $k$-stars LDP algorithm outperforms than edge LDP algorithm for all of the selected $(p,q)$-cliques, from $(2,2)$-clique to $(3,3)$-clique. Because $k$-stars LDP algorithm utilizes the structure information within $(p,q)$-cliques, which is $k$-stars, $k$-stars LDP algorithm requires far less noise than edge LDP algorithm, and obtains a better performance than it. Also, we notice that the gap of relative error between edge LDP and $k$-stars LDP increases as p and q increase. This is illustrated by the variance upper bounds of Theorem \ref{theo:utility}. The larger the p and q are, the larger the gap is. The results illustrate the effectiveness and robustness of $k$-stars LDP algorithm for different $(p,q)$-cliques, especially for the complex ones.\\

\noindent
\textbf{The effectiveness of absolute value correction technique.} As is illustrated in Section \ref{sec:k-stars LDP}, absolute value correction technique is to solve the negative-output problem. In the experiment, we set $\epsilon = 0.1$, $k=p=q=2$ and $n=10^3$.
\begin{figure}[htbp]
	\centering
	 \includegraphics[scale=0.2]{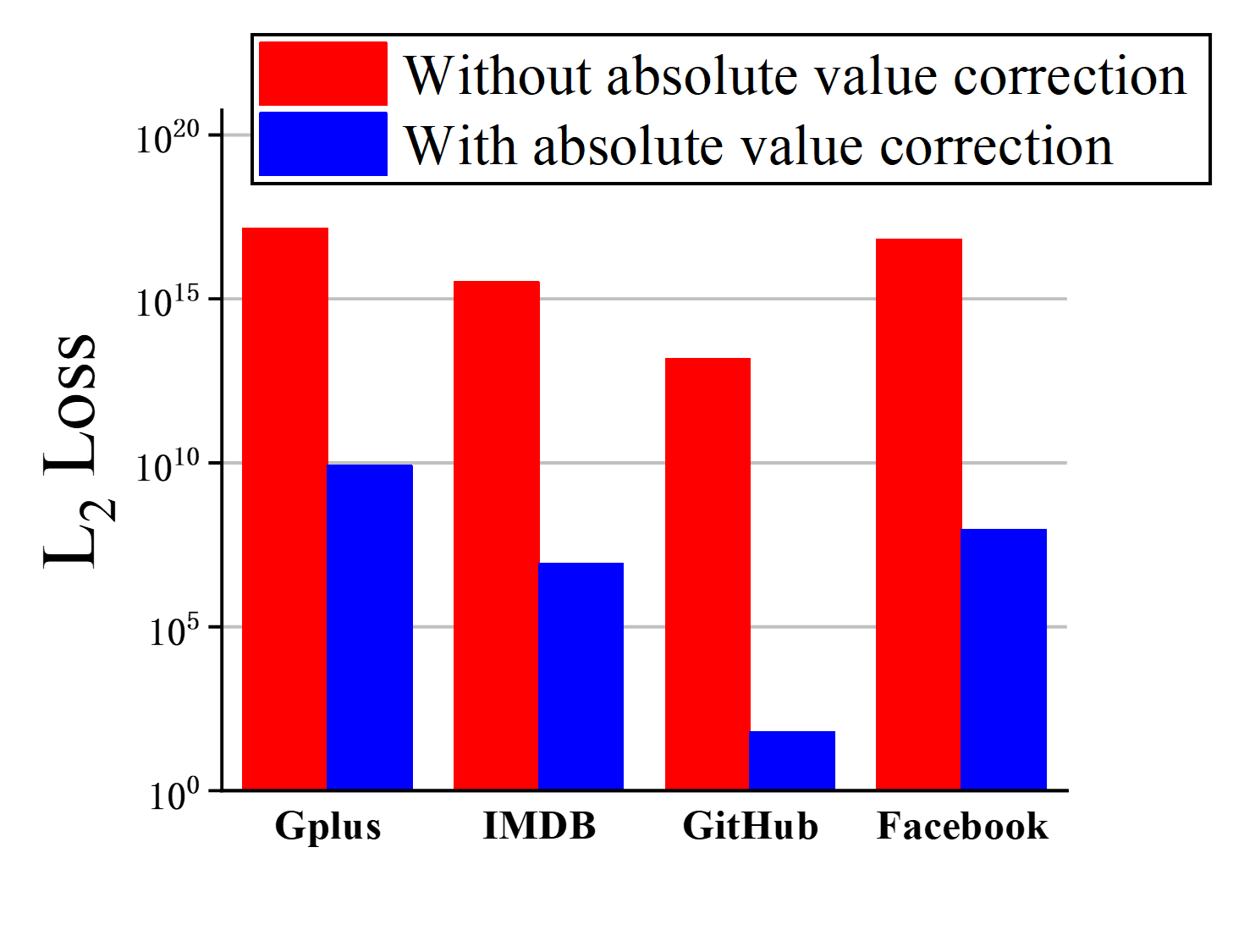}
	\caption{$L_2$ loss of algorithm with/without absolute value correction technique. ($n=10^3,\epsilon =0.1,k=p=q=2$).}
    \label{fig:absolute}
\end{figure}
It can be seen from Fig. \ref{fig:absolute} that our absolute value correction technique greatly reduced the $L_2$ loss of the algorithms, with the performance enhancement over $10^7$ order of magnitudes. Since the graph is very sparse, the algorithms tend to output negative value, due to the unbiased correction term. With the absolute value correction technique, the constraints of subgraph count problem are assured, thus leading to a better data utility.\\

\noindent
\textbf{The effectiveness of $k$-stars sampling.} As Section \ref{sec:k-stars LDP} illustrated, we borrow the basic idea of edge sampling and propose the $k$-stars sampling. In the experiment, we set $\epsilon=0.1$, $k=p=q=2$ and the sampling ratio $\rho =0.9$. Fig. \ref{fig:k-stars-sampling} shows the results.
\begin{figure}[htbp]
	\centering
	\includegraphics[scale=0.37]{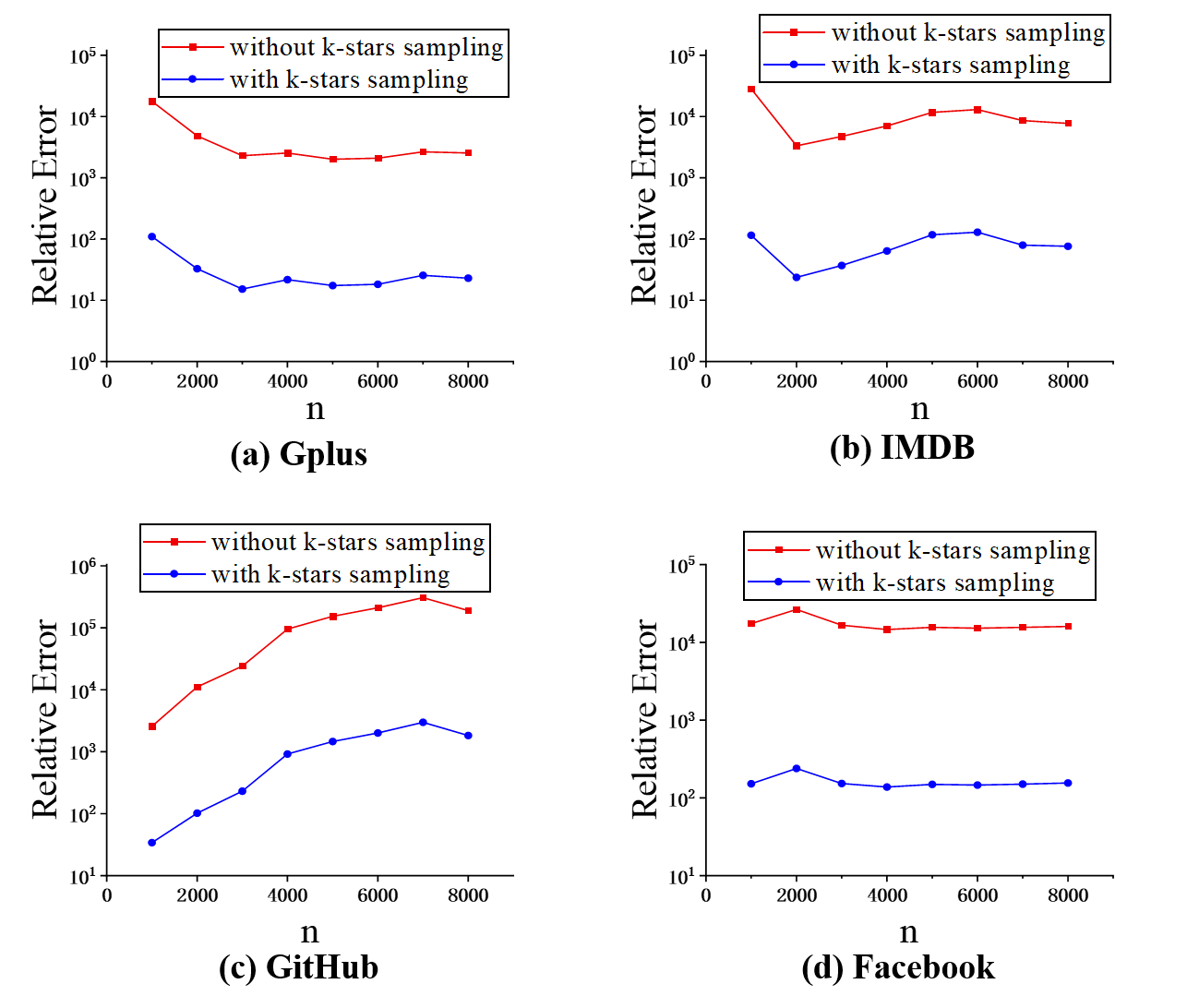}
	\caption{Relative error of $k$-stars LDP with/without $k$-stars sampling ($\epsilon =0.1,\rho=0.9, k=p=q=2$).}
    \label{fig:k-stars-sampling}
\end{figure}
It can be seen from Fig. \ref{fig:k-stars-sampling} that $k$-stars sampling can greatly reduce the relative error about 3 orders of magnitude. The RR mechanism in $k$-stars LDP tends to produces more noisy $k$-stars than the original graph. Thus, $k$-stars sampling can reduce the redundant noisy $k$-stars and improve the utility.\\

\noindent
\textbf{Increasing the privacy budget of edge LDP.} As illustrated in Theorem \ref{theo:privacy}, an $\epsilon-k$-stars LDP algorithm can provide $\frac{k2^{(k-1)}}{2^k-1}\epsilon -$edge LDP. One may wonder whether the proposed $\epsilon-k$-stars LDP algorithm still provides a better utility than the edge LDP algorithm whose privacy budget is set to $\frac{k2^{(k-1)}}{2^k-1}\epsilon$. We evaluate and compare the performances under such two circumstances. The privacy budgets are set to 0.133 for the edge LDP algorithm and 0.1 for the k-stars LDP algorithm while the other parameters remain the same as $k=p=q=2$. The results of relative error are shown in Fig. \ref{fig:re_increasing-e}, which reveal that our $k$-stars LDP algorithm still outperforms the edge LDP algorithm in all instances, despite the increased privacy budget for the edge LDP algorithm. To explain the results, $k$-stars LDP only requires one noisy wedge instead of two noisy edges. Thus, $k$-stars LDP significantly reduces the amount of required noisy and enhances the utility.\\
\begin{figure}[htbp]
	\centering
	\includegraphics[scale=0.37]{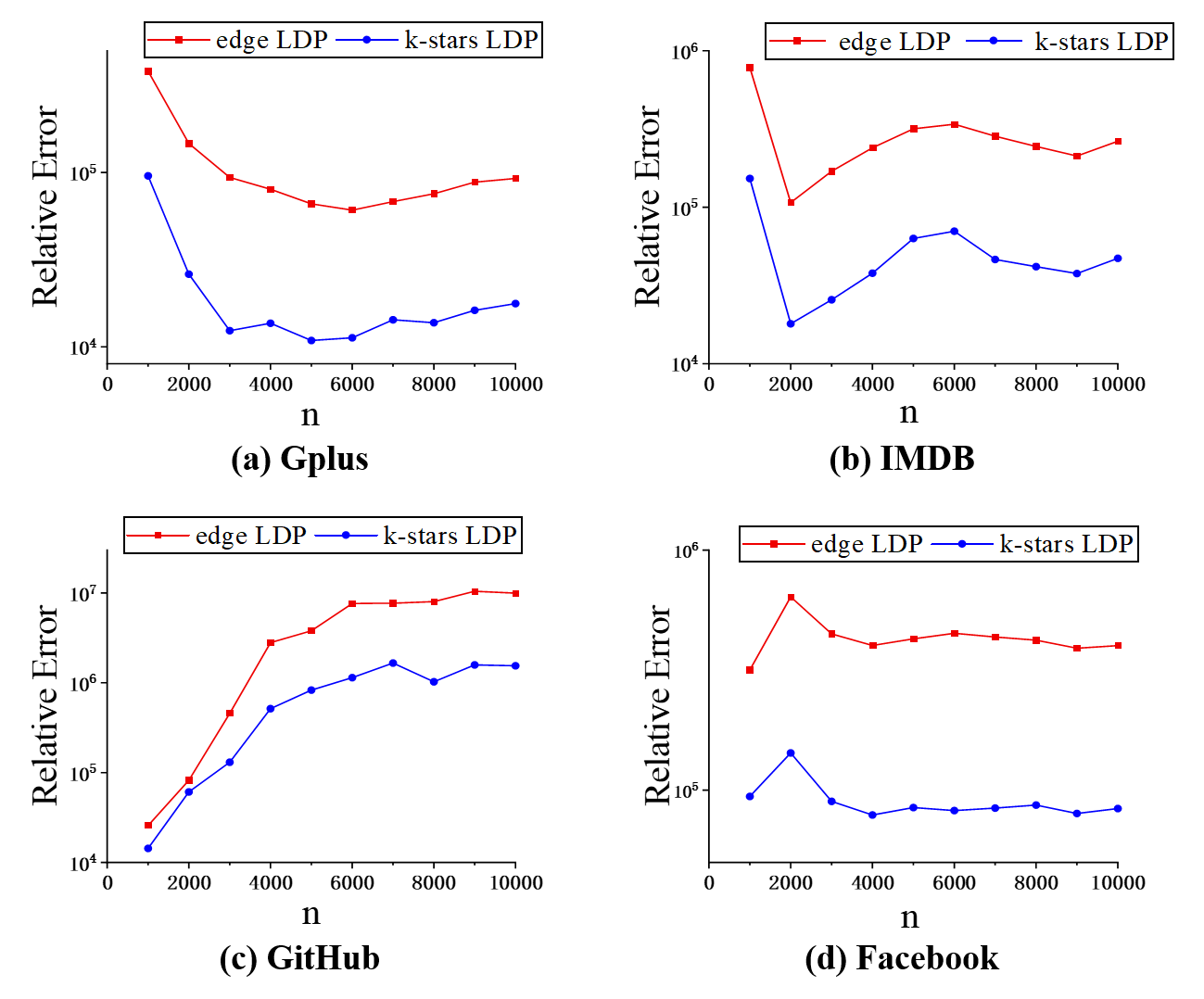}	
	\caption{Relative error of edge LDP and $k$-stars LDP with various n ($\epsilon =0.133, 0.1$ for edge LDP and $k$-stars LDP respectively, $k=p=q=2$).}
    \label{fig:re_increasing-e}
\end{figure}

\noindent
\textbf{Summary.} In summary, our $k$-stars LDP algorithm outperforms traditional edge LDP algorithm in every aspect evaluated in the experiments. When setting the same privacy budget, the $k$-stars LDP algorithm improves the utility performance over traditional edge LDP algorithm. Moreover, the superiority of $k$-stars LDP algorithm holds even if we increase the privacy budget of edge LDP algorithm. $K$-stars LDP algorithm works better in a sparse graph than in a dense graph, which means it can work better in the practical scenario as the real social network tends to be sparse. Also, the superiority of $k$-stars LDP algorithm is expanded as $(p,q)$-clique becomes more and more complex. Meanwhile, the effectiveness of $k$-stars sampling and absolute value correction is evaluated as well.
\section{Conclusions}
\label{sec:con}

In this paper, we propose the novel idea of $k$-stars LDP and the novel $k$-stars LDP algorithm for $(p,q)$-clique enumeration with a small estimation error. Based on the fact that the effectiveness of traditional edge LDP algorithm lies in the capability of obfuscating the existence of $k$-stars, we propose the $k$-stars neighboring list to enable our $k$-stars LDP algorithm to flip the 0/1 in the list. With Warner's RR and two-round user-collector interaction mechanism, the $k$-stars LDP algorithm for $(p,q)$-clique enumeration is designed to count $(p,q)$-clique under local differential privacy. Moreover, we propose the absolute value correction technique and the $k$-stars sampling technique to reduce the estimation error of the algorithm. Both theoretical analysis and experiments have proved the superiority of the proposed $k$-stars LDP algorithm over the traditional edge LDP algorithm, which shows that our $k$-stars LDP algorithm can greatly enhance the data utility while sacrificing a relatively small privacy budget.

%\clearpage

%\clearpage
\bibliographystyle{IEEEtran}
%\balance
\bibliography{reference}
%\newpage
%\vspace{-3 cm}
\begin{IEEEbiography}[{\includegraphics
[width=1in,height=1.25in,clip,
keepaspectratio]{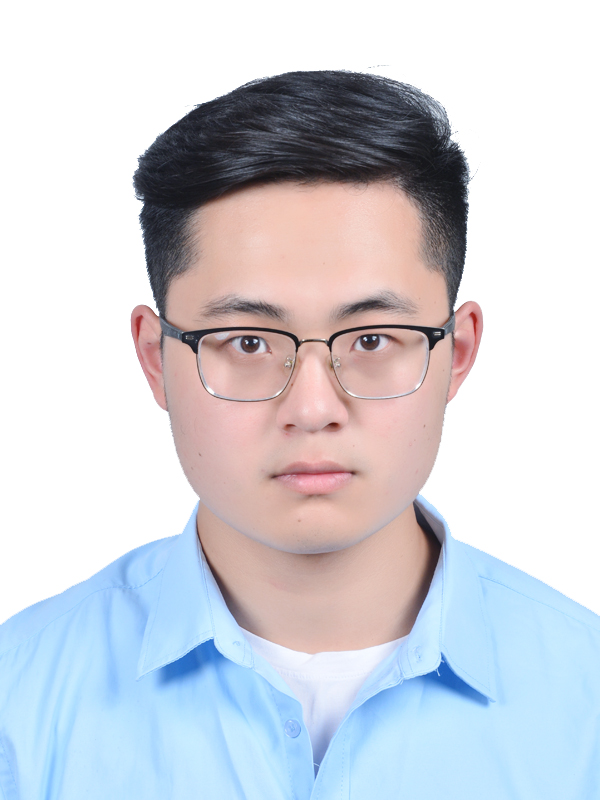}}]
{Henan Sun}
is currently a MS Candidate in Beijing Institute of Technology, China. He received his BS degree from the College of Information Science, Northwest A\&F University in 2022. His research interests include data privacy and graph neural networks.
\end{IEEEbiography}
\vspace{-0.5 cm}
\begin{IEEEbiography}[{\includegraphics
[width=1in,height=1.25in,clip,
keepaspectratio]{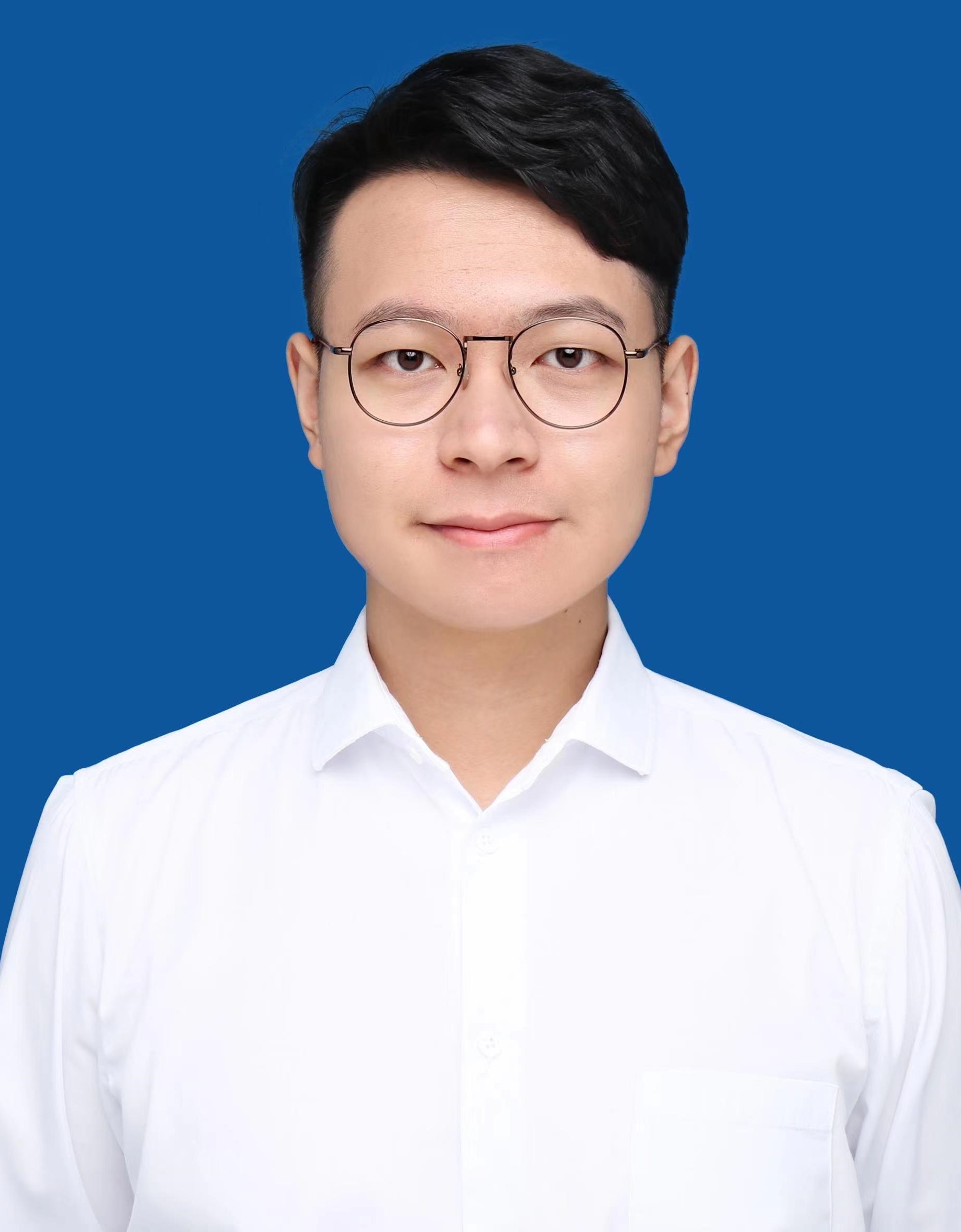}}]
{Zhengyu Wu}
is currently a PhD Candidate in Beijing Institute of Technology, China. He received his Bachelor and Master degree in International Relations from Brandeis University and Johns Hopkins University Paul H. Nitze School of Advanced International Studies(SAIS), respectively. His current research interests include Graph Neural Network, Federated Learning, and Machine Unlearning.
\end{IEEEbiography}

\begin{IEEEbiography}[{\includegraphics
[width=1in,height=1.25in,clip,
keepaspectratio]{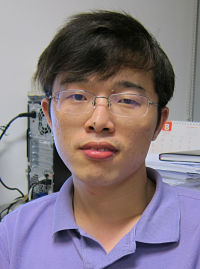}}]
{Rong-Hua Li}
received his PhD degree from the Department of Systems Engineering and Engineering Management, The Chinese University of Hong Kong in 2013 and joined the School of Computer and Software, Shenzhen University this year. He joined the School of Computer Science \& Technology, Beijing Institute of Technology in 2018 as a professor. His research interests include graph theory, graph computing system, spectral graph theory, graph neural networks and knowledge graph. He is the committee member of VLDB 2024, KDD 2023, WWW 2023.
\end{IEEEbiography}

\begin{IEEEbiography}[{\includegraphics
[width=1in,height=1.25in,clip,
keepaspectratio]{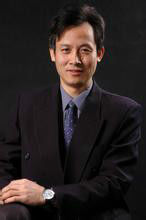}}]
{Guoren Wang}
received his MS, BS and PhD degrees from Northeastern University in 1988, 1991 and 1996. He joined the School of Computer Science \& Technology, Beijing Institute of Technology in 2018 as a professor. He is the member of the Discipline Evaluation Group of the State Council and the Expert Review Group of the Information Science Department of the National Natural Science Foundation of China. His research interests include uncertain data management, data intensive computing, visual media data analysis, and bioinformatics.
\end{IEEEbiography}

\begin{IEEEbiography}[{\includegraphics
[width=1in,height=1.25in,clip,
keepaspectratio]{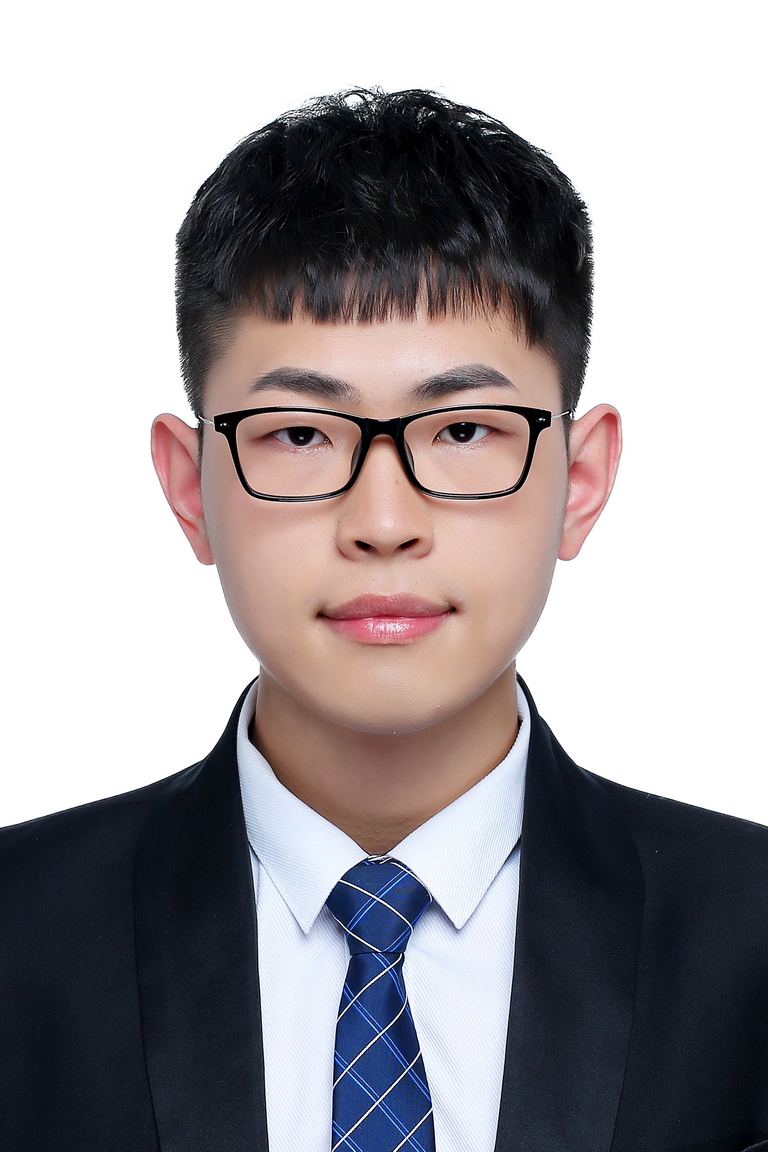}}]
{Zening Li}
is currently a PhD Candidate in Beijing Institute of Technology, China. His current research interests include differential privacy and social network analysis.
\end{IEEEbiography}

\end{document}